\documentclass[lettersize,journal]{IEEEtran}

\usepackage{amsmath,amsfonts}
\usepackage{array}
\usepackage{textcomp}
\usepackage{stfloats}
\usepackage{url}
\usepackage{verbatim}
\usepackage{graphicx}
\usepackage{cite}
\usepackage{graphicx}          
\usepackage{amsmath} 
\usepackage{amsthm}
\usepackage{amssymb} 
\usepackage{bm}
\usepackage{mathtools}
\usepackage{upgreek}
\usepackage{accents}
\mathtoolsset{showonlyrefs}
\usepackage[utf8]{inputenc}
\usepackage{amsfonts}
\usepackage{dsfont}
\usepackage{upgreek}
\usepackage[dvipsnames]{xcolor}
\usepackage{nameref} 
\usepackage{xr}
\usepackage{pifont}

\usepackage{tikz}
\usepackage{enumerate}
\usepackage{algorithm}
\usepackage{algpseudocode}
\usepackage{verbatim}
\usepackage{multirow}
\usepackage{todonotes}
\usepackage{subfig}

\newcommand{\nd}{\text{d}}
\newcommand{\mf}{\mathbf}

\newcommand{\sigmoid}{\text{\normalfont  sigmoid}}
\newcommand{\ubar}[1]{\underaccent{\bar}{#1}}
\newcommand{\sgn}[2]{\left\lceil#1\right\rfloor^{#2}}

\newtheorem{proposition}{Proposition}

\newtheorem{definition}{Definition}
\newtheorem{corollary}{Corollary}

\newtheorem{remark}{Remark}
\newtheorem{theorem}{Theorem}

\begin{document}

\title{NN-ETM: Enabling safe neural network-based event-triggering mechanisms for consensus problems}

\author{Irene Perez-Salesa,~\IEEEmembership{Graduate~Student~Member}, Rodrigo Aldana-López,~\IEEEmembership{Member}, \\ Carlos Sagüés,~\IEEEmembership{Senior~Member}
\thanks{The authors are with Departamento de Informática e Ingeniería de Sistemas (DIIS) and Instituto de Investigación en Ingeniería de Aragón (I3A), Universidad de Zaragoza, María de Luna 1, 50018, Zaragoza, Spain. 

This work was supported via projects PID2021-124137OB-I00 and PID2024-159279OB-I00 funded by MICIU/AEI/10.13039/501100011033 and by ERDF/EU, via project REMAIN S1/1.1/E0111 (Interreg Sudoe Programme, ERDF), by the Gobierno de Aragón under Project DGA T45-23R, and by Spanish grant FPU20/03134. Grant reference BG24/00121 funded by MICIU/AEI/10.13039/501100011033. }
\thanks{\color{red}This is the accepted version of the manuscript: I. Perez-Salesa, R. Aldana-López, C. Sagüés, ``NN-ETM: Enabling safe neural network-based event-triggering mechanisms for consensus problems", accepted for publication in IEEE Transactions on Systems, Man and Cybernetics: Systems, 2026, with DOI: 10.1109/TSMC.2026.3740259. \textbf{Please cite the publisher's version.}}}

\markboth{IEEE Transactions on Systems, Man and Cybernetics: Systems,~Vol.~X, No.~X, Month~202X}%
{Perez-Salesa \MakeLowercase{\textit{et al.}}: NN-ETM: Enabling safe neural network-based event-triggering mechanisms for consensus problems}


\maketitle

\begin{abstract}
Event-triggering mechanisms (ETMs) have been developed for consensus problems to reduce communication while ensuring performance guarantees, but their design has grown increasingly complex by incorporating the agent's local and neighbor information. This typically results in ad-hoc solutions, which may only work for the consensus protocol under consideration. We aim to safely incorporate neural networks in the ETM to provide a general solution while guaranteeing performance. To decouple the stability analysis of the consensus protocol from the abstraction of the neural network, we derive design criteria for the consensus and ETM pair, allowing independent analysis of each element under mild constraints. Then, we propose NN-ETM, a novel ETM featuring a neural network, {to reduce communication adaptively thanks to the variable threshold learned by the network} while preserving the stability guarantees of the consensus protocol.
\end{abstract}


\section{Introduction}
\IEEEPARstart{D}{ecentralized} {consensus problems have been widely studied in recent years due to their interest in the systems engineering field, with applications covering multi-robot control and cooperation \cite{OlfatiSaber2007, Ren2006}, smart power grids \cite{Molzahn2017}, distributed estimation and information fusion in sensor networks \cite{ifac2023, ecc24} or distributed optimization in large-scale computing systems \cite{Tsianos2012}.
Generally, consensus algorithms require frequent communication among agents in order to achieve the desired performance, which can be prohibitive in resource-constrained applications with shared transmission networks, as in industrial contexts \cite{ArandaEscolastico2024}, or with autonomous vehicles with limited battery life \cite{Xiang2023}}. 

To alleviate the communication load, event-triggering mechanisms (ETMs) have been proposed to reduce the number of necessary transmissions in networked and multi-agent systems, and their application to consensus problems has received considerable attention \cite{Ding2018}. In event-triggered schemes, transmissions are only performed when a certain condition is fulfilled. This condition is carefully designed to ensure a performance guarantee for the algorithms despite the reduction in communication \cite{Peng2018}.
{Different approaches have been taken in the literature to design suitable ETMs.}

{\textit{Fixed-threshold or static ETMs:} These ETMs feature an event-triggering condition that remains static over time. A prominent example that has been widely used for different applications is the send-on-delta trigger \cite{Miskowicz2014}, where events are triggered once the difference between the current and latest transmitted value of the signal exceeds a constant threshold. This approach is general and easy to implement, and variations exist to address other issues, e.g., quantized signals \cite{Liu2023a}. However, the threshold does not adapt according to the evolution of the signals, which limits its performance.}

{\textit{Dynamic ETMs:} More recently, DETMs have been developed, which incorporate an auxiliary dynamic variable to reduce the number of events further, in an adaptive fashion \cite{Girard2015}. DETMs have been exploited in various control and estimation problems \cite{Ge2021, Lui2021} and extended to multi-agent setups \cite{Ge2020}. 
{Particularly, several works consider the application of DETMs for static consensus problems \cite{Liu2023,Wang2025}, dynamic average consensus \cite{George2018,Qian2023,Xu2024}, leader-follower control \cite{Du2020,Zhang2022,Zhang2023detm,Li2025,Cui2025,Lu2026}, leaderless consensus control \cite{Hu2020,Zhao2021,Cao2022}, formation-containment control \cite{Zhang2024}, cluster consensus \cite{Zhou2025} and resilient consensus under unreliable conditions \cite{Wen2022,Lin2024,He2024,Yang2024,Jia2024}.}
Note that, in a multi-agent setup, the amount of information each agent can exploit to construct a triggering condition is greater than in the single-agent case since the information shared by neighboring agents can also be used. This leads to an increasingly complex design of DETMs: the auxiliary variable's dynamics can depend jointly on the event-triggered error and the disagreement with neighbors, for example. Advantageously combining these pieces of information is not trivial, as can be observed in the referenced works. Moreover, the DETM design is typically tailor-made for a particular consensus algorithm, providing a rigid ad-hoc solution that may not generalize to other cases.}

{\textit{Learning-based methods:} Unlike hand-crafted approaches, data-driven methods such as Neural Networks (NNs) have been used to work around standard feature engineering in favor of automated design, achieving less conservative results. The flexibility of NNs is currently being exploited for event-triggered control to learn complex system dynamics or optimal control policies \cite{Sedghi2022}. A few works also address the problem of learning a communication policy jointly with the control actions {\cite{Kesper2023,Shibata2021,Shibata2023,Dang2024}}. The main drawback of incorporating this kind of technique is the difficulty in providing stability guarantees for the setup due to the abstraction introduced by NNs. {In \cite{Persis2023, Xu2024datadetm}, data-driven event-triggered control with formal guarantees is studied, but the approach is limited to choosing the values of constant parameters in the hand-crafted triggering condition from data.} To the best of our knowledge, there are no works focusing on learning suitable dynamic event-triggering conditions for consensus problems.}

In addition, there are reasons to advocate for more general frameworks in which solutions may be applied outside the concrete problems and assumptions for which they were designed. In this context, systems are often divided into interconnected elements that must work together in a control stack, so hierarchical approaches that decouple the design of each block can be beneficial \cite{Dorfler2024,Wang2024}. Such a decoupled approach could also be of interest for event-triggered consensus, motivating a block-wise design strategy that complies with the requirements to ensure boundedness of the consensus error, regardless of the particular characteristics of each block.

{Motivated by this discussion, we propose NN-ETM, a novel ETM featuring a NN. {Compared to fixed-threshold ETMs, the NN adaptively tunes the event threshold to balance accuracy and communication. However, in contrast to existing hand-crafted DETMs, which are specifically tailored to one consensus protocol, the proposal provides a solution that generalizes over a wider class of protocols by leveraging NN methods. In addition, we provide a control-theoretic design framework to ensure a desired worst-case bound on the consensus error in practice, as opposed to existing learning-based approaches with no formal guarantees.
To the best of our knowledge, this is the first work that safely incorporates NNs in the triggering condition to replace hand-crafted DETMs.} We outline a training process for NN-ETM, including practical considerations, and we provide publicly available code for training and testing NN-ETM.}

{The paper is organized as follows. Section \ref{sec:prob-stat} introduces the problem statement. In Section \ref{sec:design-criteria}, design conditions for stability are derived. Section \ref{sec:nn-etm} introduces NN-ETM and discusses the NN inputs, architecture and training. In Section \ref{sec:iss-consensus}, we include examples illustrating the validation of stability requirements and comments on parameter choice. Section \ref{sec:discussion} discusses the generalizability of our framework and the comparison with existing approaches. Section \ref{sec:experiments} showcases experiments across different scenarios. The conclusions are given in Section \ref{sec:conclusions}.}

\subsection{Notation}
Let the $n \times n$ identity matrix be denoted as $\mf{I}_n$ and $\mathds{1}_n=[1,\dots,1]^\top$ with dimension $n \times 1$. The operator $\| \bullet \|$ denotes the Euclidean norm. Let $\text{sign}(x) = 1$ if $x> 0$, $\text{sign}(x)=-1$ if $x<0$ and $\text{sign}(0)=0$. We use the same notation for the element-wise application of this operator to vectors or matrices. Let $\otimes$ denote the Kronecker product. We use the classes of comparison functions $\mathcal{K}_\infty,\mathcal{KL}$ from \cite[Page 144]{Khalil2002}. We use $\lambda_{\max}(\bullet)$ to denote the maximum eigenvalue of a matrix.

\section{Problem Statement}\label{sec:prob-stat}

Consider a network of $N$ interconnected agents. The network is described by an undirected connected graph $\mathcal{G}=(\mathcal{V}, \mathcal{E})$, where the agents form the vertex set $\mathcal{V}=\{ 1, \dots, N \}$ and the edge set $\mathcal{E} \subseteq \mathcal{V} \times \mathcal{V}$ represents the communication links between agents. The network topology is characterized by its adjacency matrix $\mf{A}_\mathcal{G}$, equivalently by its Laplacian $\mf{Q}_\mathcal{G} = \text{diag}(\mf{A}_\mathcal{G}\mathds{1})-\mf{A}_\mathcal{G}$ and incidence matrix $\mf{D}_{\mathcal{G}}$ satisfying $\mf{Q}_{\mathcal{G}}=\mf{D}_{\mathcal{G}}\mf{D}_{\mathcal{G}}^\top$. Moreover, $\lambda_2(\mathcal{G})$ denotes the algebraic connectivity of the graph, i.e., the smallest non-zero eigenvalue of $\mf{Q}_{\mathcal{G}}$. The neighbors of an arbitrary agent $i\in\mathcal{V}$ are denoted by $\mathcal{N}_i\subseteq\mathcal{V}$ and $\mathcal{N}'_i=\mathcal{N}_i\cup\{i\}$.

Several consensus solutions can be written in terms of differential equations of the form
\begin{equation}\label{eq:consensus}
    \dot{\mf{z}}_i(t) = \mf{f}_i(t, \mf{z}_i(t), \{ \mf{m}_j(\mf{z}_j(t)) \}_{j \in \mathcal{N}_i'})
\end{equation}
for each agent $i$ and some appropriate $\mf{f}_i(\bullet)$, {which in practice may correspond to the actual dynamics of a physical system with state $\mf{z}_i(t)\in\mathbb{R}^n$, or the description of an iterative algorithm executed on a computing platform \cite{Dorfler2024}.}
In \eqref{eq:consensus}, local interaction between agents occurs due to the terms $\{ \mf{m}_j(\mf{z}_j(t)) \}_{j \in \mathcal{N}_i'}$ which correspond to information that is communicated between neighbors, for some appropriate output message function $\mf{m}_j(\bullet)$. The goal for the agents is to reach an agreement as $
    \lim_{t\to\infty} \| \mf{z}_i(t) - \mf{z}_j(t) \| = 0, \forall i,j \in \mathcal{V}.
$
The simplest example is $\mf{f}_i(t,\mf{z}_i,\{\mf{m}_j(\bullet)\}_{j\in\mathcal{N}_i'}) = -\sum_{j\in\mathcal{N}_i}(\mf{z}_i-\mf{m}_j(\bullet))$, with $\mf{m}_j(\bullet) = (\bullet)$, resulting in the standard consensus iteration $
\dot{\mf{z}}_i(t) = -\textstyle\sum_{j\in\mathcal{N}_i}(\mf{z}_i(t)-\mf{z}_j(t)).$

In event-triggered consensus, each agent evaluates a local event-triggering condition to decide when to communicate with its neighbors. Thus, agent $i$ only has access to $\mf{m}_j(\mf{z}_j(t)), j\in\mathcal{N}_i$ at event instants $t \in \{ \tau_k^j \}_{k=0}^\infty$, with $\tau_k^j$ being the $k$th event triggered by agent $j$. In these conditions, \eqref{eq:consensus} is replaced by
\begin{equation}\label{eq:consensus-ev}
    \dot{\mf{z}}_i(t) = \mf{f}_i(t, \mf{z}_i(t), \{ \mf{m}_j(\mf{z}_j(\tau_t^j)) \}_{j \in \mathcal{N}_i'}), 
\end{equation}
where $\tau_t^j = \max \{\tau_k^j \leq t \}$ represents the last event triggered by agent $j$, at which it broadcast its state. 

In the following, we aim to design a NN-ETM for consensus problems. Note that the introduction of NNs typically hinders formal analysis, due to their high level of abstraction. Thus, we start by proposing some design criteria for the ETM and consensus pair, which allow decoupling the design for the ETM from the consensus protocol. This decoupling is necessary in order to analyze the stability of the consensus protocol under the ETM without relating it to the NN itself. Using the proposed criteria, we design and test the NN-ETM to ensure the boundedness of the consensus error.

\section{Design Criteria for ETM in Consensus}\label{sec:design-criteria}

\subsection{Input-to-State Stability}

Before introducing our framework to design the NN-ETM, recall the definition of input-to-state (practical) stability, adapted from \cite{Mironchenko2019}:

\begin{definition}\label{def:iss}
A system $\dot{\mf{x}}(t) = \mf{h}(t, \mf{x}(t), \mf{u}(t))\in\mathbb{R}^n$ is said to be uniformly \emph{input-to-state practically stable} (ISpS) with respect to the origin
if there exist a class $\mathcal{KL}$ function $\beta$, a class $\mathcal{K}_\infty$ function $\gamma$ and a constant $c$ such that, for any initial state $\mf{x}(0)$ and bounded input $\mf{u}(t)$, the solution $\mf{x}(t)$ satisfies
\begin{equation}\label{eq:iss}
    \| \mf{x}(t) \| \leq \beta(\| \mf{x}(0) \|, t) + \gamma \textstyle\left(\sup_{s\in[0, t]} \| \mf{u}(s) \|\right)  + c.
\end{equation}
If the ISpS property holds for $c=0$, then the system is \emph{input-to-state stable} (ISS).
\end{definition}

{The ISS property implies a notion of robustness with respect to the input $\mf{u}(t)$, which may be a disturbance. In our case, we are interested in studying the ISS property of the disagreement dynamics for consensus with respect to the disturbance introduced by event-triggered communication.}

\subsection{Interconnection of ISS consensus and ETM}
\label{sec:design:criteria}

\begin{figure}
   \centering
   \includegraphics[width=0.65\columnwidth]{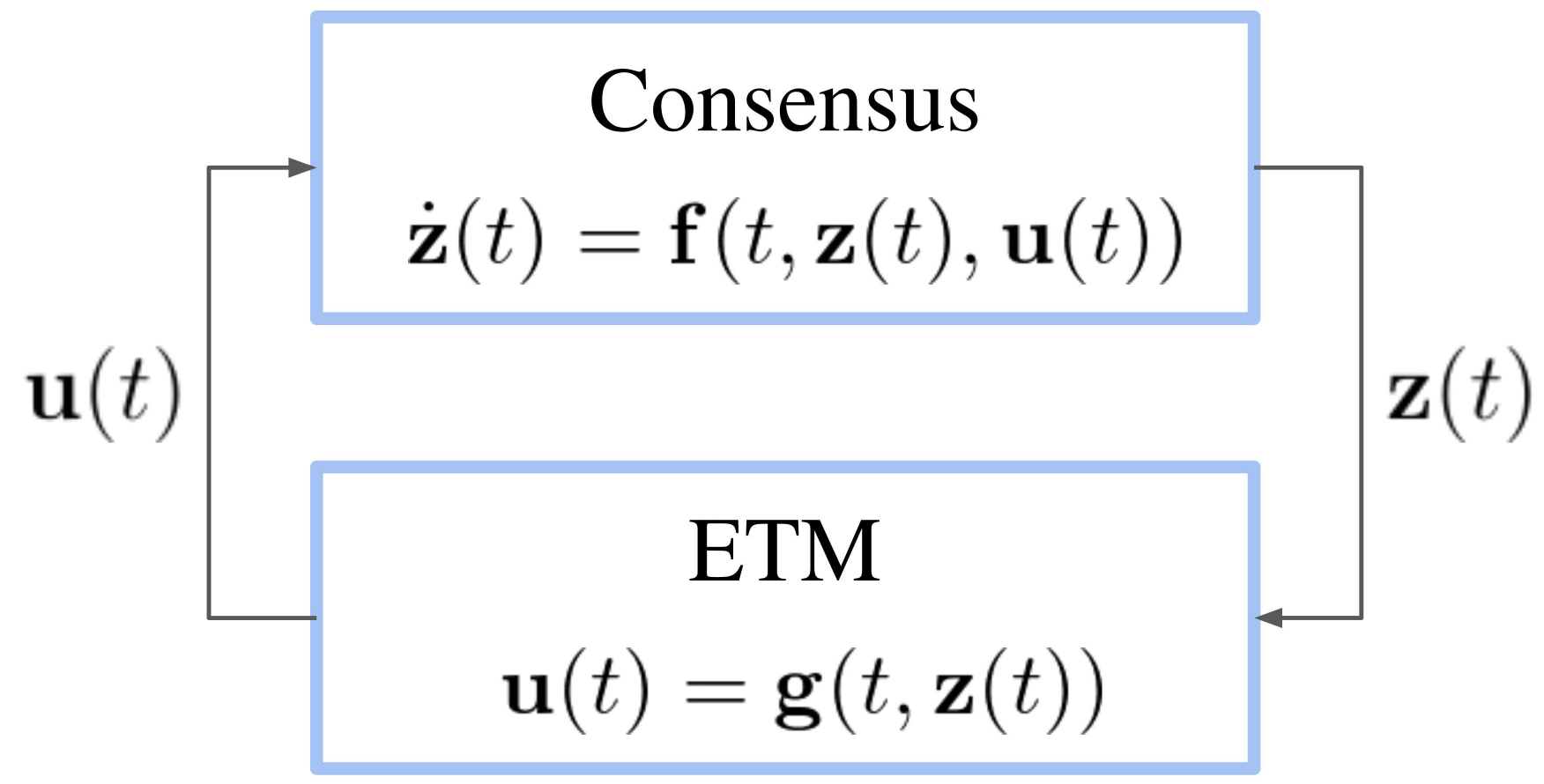}
   \caption{The event-triggered consensus problem can be globally viewed as the interaction between two interconnected blocks. We derive design criteria for each block to ensure stability guarantees for the interconnection.}
    \label{fig:blocks}
\end{figure}

We now express the event-triggered consensus problem as a two-block interconnection. Consider system \eqref{eq:consensus} with continuous-time communication, written in compact form as $\dot{\mf{z}}(t) = \mf{f}'(t,\mf{z}(t), \mf{m}(\mf{z}(t)))$ by defining $\mf{z}(t) = [\mf{z}_1(t)^\top, \dots, \mf{z}_N(t)^\top]^\top$, $\mf{m}(\mf{z})=[\mf{m}_1(\mf{z}_1)^\top,\dots,\mf{m}_N(\mf{z}_N)^\top]^\top$ and an appropriate $\mf{f}'(\bullet)$. Similarly, the event-triggered protocol \eqref{eq:consensus-ev} is expressed as
\begin{equation}\label{eq:consensus-ev-vec}
\begin{aligned}
\dot{\mf{z}}(t) &= \mf{f}'(t,\mf{z}(t), \mf{m}(\mf{z}(t)) + \mf{u}(t)) =:\mf{f}(t,\mf{z}(t),\mf{u}(t)),
\end{aligned}
\end{equation}
where $\mf{u}(t) = [\mf{u}_1(t)^\top, \dots, \mf{u}_N(t)^\top ]^\top$ and $\mf{u}_j(t) = \mf{m}_j(\mf{z}_j(\tau_t^j))-\mf{m}_j(\mf{z}_j(t))$ describes the mismatch between the current values of $\mf{m}_j(\mf{z}_j(t))$ and the actual transmitted information $\mf{m}_j(\mf{z}_j(\tau_t^j))$. Therefore, the combined effect of the ETMs at the agents is written as a block 
\begin{equation}\label{eq:disturb-block}
\mf{u}(t) = \mf{g} (t, \mf{z}(t)),
\end{equation}
for $\mf{g}=[\mf{g}_1^\top,\dots,\mf{g}_N^\top]^\top$, with $\mf{g}_j(t,\mf{z}(t)) = \mf{m}_j(\mf{z}_j(\tau_t^j))-\mf{m}_j(\mf{z}_j(t)).$ The interconnection of the blocks \eqref{eq:consensus-ev-vec} and \eqref{eq:disturb-block} is illustrated in Figure \ref{fig:blocks}.

Given that the goal of the consensus protocol is to achieve agreement between the agents of the network, convergence is typically analyzed through the dynamics for the disagreement. Let $\bar{\mf{z}}(t)$ be the desired agreement value, e.g., in average consensus we have $\bar{\mf{z}}(t) = (1/N) \sum_{i=1}^N \mf{z}_i(t)$. We can split 
\begin{equation}\label{eq:split-consensus-disagr}
\mf{z}(t) = (\mathds{1}_N \otimes \mf{I}_n) \bar{\mf{z}}(t) + \mf{x}(t),
\end{equation}
where $\mf{x}(t)$ is the disagreement vector, defined as $\mf{x}(t) = (\mf{H} \otimes \mf{I}_n) \mf{z}(t)$ with $\mf{H} = \mf{I}_N - (1/N)\mathds{1}_N\mathds{1}_N^\top$ being the projection matrix in the space orthogonal to consensus. Then, the dynamics for the disagreements are given by
\begin{equation}\label{eq:error-dyn}
    \dot{\mf{x}}(t) = \mf{h}(t, \mf{x}(t), \mf{u}(t)) := (\mf{H} \otimes \mf{I}_n) \dot{\mf{z}}(t).
\end{equation}
Note that $\mf{x}(t)=\mf{0}$ for some $t\in\mathbb{R}$ means that consensus $\mf{z}_i(t)=\mf{z}_j(t)$ for all $i,j\in\mathcal{V}$ is achieved at that instant.

In the following theorem, we summarize general design criteria for the ETM and consensus blocks in order to guarantee the boundedness of the consensus error under event-triggered communication. Subsequently, we provide additional results to show general cases of designs where these criteria can be satisfied.

\begin{theorem}[Design Criteria]\label{th:iss}
Consider an event-triggered consensus protocol \eqref{eq:consensus-ev-vec} with disagreement dynamics \eqref{eq:error-dyn}. Then, there exists $\xi\geq 0$ such that for all initial conditions $\mf{x}(0)$, the consensus error is bounded $\forall t\geq 0$ with $\limsup_{t\to\infty} \|\mf{x}(t)\| \leq \xi$, if the following design conditions are fulfilled:
\begin{enumerate}[(i)]
    \item\label{it:existence} The functions $\{\mf{f}_i,\mf{m}_i\}_{i\in\mathcal{V}}$ ensure existence of solutions of \eqref{eq:consensus} for all $t\geq 0$.
    \item \label{it:zeno} The ETM guarantees a minimum inter-event time.
    \item \label{it:iss} The disagreement dynamics \eqref{eq:error-dyn} are ISpS.
    \item \label{it:bounded-u} The disturbance $\mf{u}(t)$ produced by the ETM has a uniformly bounded Euclidean norm for all $t\geq 0$.
\end{enumerate}
\end{theorem}
\begin{proof}
First, we show that trajectories $\mf{x}(t)$ exist for all $t\geq 0$. Item (\ref{it:existence}) ensures the existence of solutions between events. In addition, due to the ETM, the existence of solutions can be lost due to Zeno behavior, i.e., triggering an infinite number of events in a finite time interval. That is prevented by item (\ref{it:zeno}), which ensures the absence of Zeno behavior so that $\lim_{k\to\infty}\tau_k^i=\infty$ for all $i\in\mathcal{V}$ \cite{Yu2021}. Hence, solutions of \eqref{eq:error-dyn} exist $\forall t\geq 0$. 
For the boundedness of the consensus error, item (\ref{it:iss}) implies $\limsup_{t\to\infty} \|\mf{x}(t)\| \leq \gamma \textstyle\left(\sup_{t\geq0} \|\mf{u}(t)\|\right) + c\leq \gamma(U) + c=:\xi$ with $U$ being the uniform bound of $\|\mf{u}(t)\|$ from item (\ref{it:bounded-u}), completing the proof.
\end{proof}

Notice that items (\ref{it:existence}) and (\ref{it:iss}) in Theorem \ref{th:iss} relate to the design of the consensus protocol, while items (\ref{it:zeno}) and (\ref{it:bounded-u}) are restrictions for the design of the ETM. 
Particularly, note that item (\ref{it:iss}) implies that the consensus protocol can be analyzed with respect to a generic communication disturbance, without assuming a particular ETM. Since the disturbance is caused by the event-triggering mechanism in our case, the ISS property means that the disagreement dynamics are robust to such disturbance. This provides freedom to introduce data-driven techniques such as NNs in the ETM under some mild constraints (namely, the boundedness of $\mf{u}(t)$), without hindering formal analysis of the consensus protocol.
These requirements can be relaxed for concrete choices of consensus protocols and ETMs. In the following proposition, we study the standard ETM choice
\begin{equation}
\label{eq:send:on:delta}
\tau_{k+1}^i = \inf\{ t > \tau_k^i : \| \mf{m}_i(\mf{z}_i(t)) - \mf{m}_i(\mf{z}_i(\tau_k^i)) \| \geq \delta(t)\} ,
\end{equation}
where $\delta(t)\geq0$ may have additional dynamics.
In this setting, item (\ref{it:bounded-u}) of Theorem \ref{th:iss} can be relaxed as follows, {noting that the disturbance term now depends on the event threshold}.

\begin{proposition}[Performance of the ETM \eqref{eq:send:on:delta}]\label{prop:iss}
Consider a consensus protocol \eqref{eq:consensus-ev-vec} satisfying items \eqref{it:existence} and \eqref{it:iss} in Theorem \ref{th:iss}, along with the ETM \eqref{eq:send:on:delta}. Then, solutions to \eqref{eq:error-dyn} fulfill
\begin{equation}\label{eq:bound-z}
    \|\mf{x}(t)\| \leq \beta(\| \mf{x}(0) \|, t) + \gamma' \textstyle\left(\sup_{s\in[0, t]}  \delta(s) \right) + c
\end{equation}
for all $t$ in which the trajectory $\mf{x}(t)$ exists, for some $\beta \in \mathcal{KL}, \, \gamma' \in \mathcal{K}_\infty$.
Moreover, if \eqref{eq:error-dyn} is ISS, $\lim_{t\to\infty} \delta(t) = 0$, and $\mf{x}(t)$ exists for all $t\geq 0$, then $\lim_{t\to\infty} \|\mf{x}(t)\| = 0$.
\end{proposition}
\begin{proof}
The design of the ETM in \eqref{eq:send:on:delta} ensures $\| \mf{u}_i(t)\| = \|\mf{m}_i(\mf{z}_i(t)) - \mf{m}_i(\mf{z}_i(\tau_t^i)) \| \leq \delta(t)$ for all $t$ in which solutions exist and therefore $\| \mf{u}(t) \| \leq \sqrt{N} \delta(t)$. Moreover, note that the ISpS property in \eqref{eq:iss} holds, from which \eqref{eq:bound-z} follows by picking $\gamma'(\bullet):=\gamma(\sqrt{N}\bullet)$. For the last part of the proposition, recall that ISS implies that $c=0$ in \eqref{eq:iss} and \eqref{eq:bound-z}, so that the steady-state consensus error is uniquely determined by the disturbance input. As a result, if solutions exist for all $t\geq 0$, then $\lim_{t\to\infty} \|\mf{x}(t)\| \leq \gamma'\left(\lim_{t\to\infty} \delta(t)\right)=0$ if $\lim_{t\to\infty}\delta(t)=0$, completing the proof.
\end{proof}

As mentioned in the previous result, convergence requires solutions to exist for all $t\geq 0$, which may fail in the presence of Zeno behavior, {i.e., if an infinite number of events accumulates in finite time}. For this purpose, the following result establishes the existence of a minimum inter-event time for \eqref{eq:send:on:delta} under very mild conditions:

\begin{proposition}[Minimum inter-event time for \eqref{eq:send:on:delta}]
\label{prop:zeno}
Consider the consensus protocol \eqref{eq:consensus-ev} with the ETM in \eqref{eq:send:on:delta} under the following conditions:
\begin{enumerate}[(i)]
    \item Items \eqref{it:existence} and \eqref{it:iss} in Theorem \ref{th:iss} are satisfied.
    \item\label{it:boundedfunctions} The functions $\mf{f}_i(\bullet)$ have bounded outputs and $\mf{m}_i(\bullet)$ are differentiable for all values of the arguments, $ \forall i\in\mathcal{V}$.
    \item\label{it:threshold} There exist 
 $0<\underline{\delta}\leq\overline{\delta}$ such that $\delta(t)\in[\underline{\delta},\overline{\delta}]$, $\forall t\geq 0$.
    \item\label{it:boundedconsensus} For every initial condition of \eqref{eq:consensus-ev}, there exists $\bar{Z}>0$ such that the consensus trajectory $\bar{\mf{z}}(t)$ satisfies $\|\bar{\mf{z}}(t)\|\leq \bar{Z}$ for all $t$ in which $\bar{\mf{z}}(t)$ exists.
\end{enumerate}
Then, every trajectory of \eqref{eq:consensus-ev} exists for all $t\geq 0$, and for each trajectory there exists $\tau>0$ such that $\tau^{i}_{k+1}-\tau^i_{k}\geq \tau$ for all $i\in\mathcal{V}, k\geq 0$. 
\end{proposition}
\begin{proof}
Consider a trajectory of \eqref{eq:consensus-ev} along with the trigger \eqref{eq:send:on:delta} for arbitrary initial conditions. Let $T\in\mathbb{R}\cup\{\infty\}$ be the greatest possible time such that $\mf{x}(t)$ exists for all $t\in[0,T)$. Then, Theorem \ref{th:iss}, together with the arguments in the first part of the proof of Proposition \ref{prop:iss}, implies that $\mf{x}(t)$ complies with the bound \eqref{eq:bound-z} for all $t\in[0,T)$. Moreover, set
$X:=\beta(\|\mf{x}(0)\|,0)+\gamma'(\overline{\delta})+c$ so that $\|\mf{x}(t)\|\leq X$ for all $t\in[0,T)$. Define $\mf{p}(t) := \mf{m}(\mf{z}(t))$ and recall \eqref{eq:split-consensus-disagr}. We have
$
\dot{\mf{p}}(t) = \frac{\partial \mf{m}}{\partial \mf{z}}((\mathds{1}_N\otimes\mf{I}_n)\dot{\bar{\mf{z}}}(t)+ \dot{\mf{x}}(t))=:\mf{q}(t,\bar{\mf{z}}(t),\mf{x}(t)),
$
where $\mf{q}(\bullet)$ is obtained from the dynamics of $\bar{\mf{z}}(t)$ and $\mf{x}(t)$ directly from \eqref{eq:consensus-ev}. Note that $\mf{q}(\bullet)$ is bounded due to item (\ref{it:boundedfunctions}). Moreover, the uniform bound for $\|\mf{x}(t)\|$ obtained before in conjunction with the uniform bound for $\|\bar{\mf{z}}(t)\|$ in item (\ref{it:boundedconsensus}) allows us to ensure the existence of $Q=\sup\{\|\mf{q}(t,\bar{\mf{z}},\mf{x})\| : t\in[0,T), \|\bar{\mf{z}}\|\leq \bar{Z}, \|\mf{x}\|\leq X\}$  such  that $\|\dot{\mf{p}}(t)\|\leq Q$ for all $t\in[0,T)$. Note that $Q$ depends on the particular trajectory of the system. Integrating over an arbitrary $[\tau_k^i,t]\subset[0,T)$, $\mf{p}(t)-\mf{p}(\tau_k^i) = \textstyle\int_{\tau_k^i}^{t}\mf{q}(s,\bar{\mf{z}}(s),\mf{x}(s))\text{d}s$. Moreover, note that 
$\|\mf{m}_i(\mf{z}_i(t))-\mf{m}_i(\mf{z}_i(\tau_k^i))\|\leq \|\mf{m}(\mf{z}(t))-\mf{m}(\mf{z}(\tau_k^i))\|=\|\mf{p}(t)-\mf{p}(\tau_k^i)\|\leq Q(t-\tau_k^i)<\underline{\delta}\leq \delta(t), $
for all $t\in[\tau_k^i,\tau_k^i+\tau)$ with $\tau=\underline{\delta}/Q>0$. Therefore, $\tau_{k+1}^i\geq \tau_k^i+\tau$, which implies $\lim_{k\to\infty}\tau_k^i=T=\infty$ for all $i\in\mathcal{V}$, completing the proof. 
\end{proof}

\begin{remark}\label{rem:zeno}
{While providing general theoretical conditions for the absence of Zeno behavior in event-triggered consensus is challenging, Proposition \ref{prop:zeno} gives mild conditions that are fulfilled by a wide class of consensus protocols and ETMs of the form \eqref{eq:send:on:delta}, such as our NN-ETM. For protocols where these conditions may not hold, e.g., in systems affected by unbounded stochastic noise, other strategies such as time regularization \cite{Ge2021} can be used, where the triggering condition is only checked after a given time $\ubar{\tau}>0$ has elapsed since the last event, ensuring an inter-event time equal to $\ubar{\tau}$. Moreover, in cyber-physical systems handled by computers, the discretized operation naturally enforces a minimum inter-event time, guaranteeing the absence of Zeno behavior.}
\end{remark}

Given the decoupled design conditions for the ETM and the consensus protocol, we are now ready to introduce our proposal for a NN-ETM. 

\section{NN-ETM: Neural Network-Based Event-Triggering Mechanism}\label{sec:nn-etm}

\subsection{NN-ETM Structure}
We propose that each agent $i$ decide its own event instants according to the ETM \eqref{eq:send:on:delta} in conjunction with
\begin{equation}\label{eq:trigger-learning}
    \delta(t)=  \sigma\eta_i(t) + \varepsilon,
\end{equation}
where the variable $\eta_i(t)\in[0,1]$ is determined by a local NN, such that an appropriate form of the trigger is learned from data, aiming to optimize the behavior of the setup according to a certain cost function. The parameters $\sigma\geq 0, \varepsilon> 0$ are user-defined constants. We refer to this strategy as NN-ETM. 

The following result summarizes the performance guarantees for a consensus algorithm under our NN-ETM:
\begin{corollary}[Guarantees for NN-ETM]\label{cor:performance-nnetm}
Consider a consensus protocol \eqref{eq:consensus-ev-vec} satisfying the conditions (\ref{it:existence}), (\ref{it:boundedfunctions}) and (\ref{it:boundedconsensus}) in Proposition \ref{prop:zeno}, along with the NN-ETM given by \eqref{eq:send:on:delta}-\eqref{eq:trigger-learning} with $\sigma \geq 0$, $\varepsilon > 0$. Then, the solutions to \eqref{eq:error-dyn} fulfill
\begin{equation}\label{eq:bound-nnetm}
    \|\mf{x}(t)\| \leq \beta(\| \mf{x}(0) \|, t) + \gamma' (\sigma + \varepsilon) + c
\end{equation}
for all $t$ in which the trajectory $\mf{x}(t)$ exists, for some $\beta \in \mathcal{KL}, \, \gamma' \in \mathcal{K}_\infty$ and constant $c \geq 0$.
Moreover, there is a minimum inter-event time $\tau > 0$ such that $\tau^{i}_{k+1}-\tau^i_{k}\geq \tau$ for all $i\in\mathcal{V}, k\geq 0$. 
\end{corollary}
\begin{proof}
The first part of the proof to show \eqref{eq:bound-nnetm} follows from Proposition \ref{prop:iss}, noting that $\delta(t) \leq \sigma + \varepsilon, \, \forall t$, by design of the threshold in NN-ETM as \eqref{eq:trigger-learning}, since $\eta_i(t) \in [0, 1]$. 
The second part follows from Proposition \ref{prop:zeno}, noting that condition \eqref{it:threshold} is fulfilled by our NN-ETM setting $\ubar{\delta} = \varepsilon>0$, and $\bar{\delta} = \sigma + \varepsilon$.
\end{proof}

{By Corollary \ref{cor:performance-nnetm}, the consensus error bound is now related to the choice of parameters $\sigma, \varepsilon$, which can be tuned by the user to guarantee a certain performance.} Note that NN-ETM fulfills item \eqref{it:bounded-u} of Theorem \ref{th:iss} by construction. In addition, item \eqref{it:zeno} is also fulfilled, according to Corollary \ref{cor:performance-nnetm}. Then, due to Theorem \ref{th:iss}, this ETM design is admissible, independent of the consensus protocol. 
Therefore, the NN-ETM can be used alongside any consensus protocol that fulfills the ISpS requirement (\ref{it:iss}), providing a generic solution to achieve a bounded consensus error with the performance guarantees established in Corollary \ref{cor:performance-nnetm}. {It could also be applied to consensus protocols with additional characteristics, such as time delays \cite{Ma2023}, as long as the conditions from Theorem \ref{th:iss} are fulfilled.}
Additionally, depending on the NN architecture and training, different behaviors may be learned for $\eta_i(t)$, exploiting the flexibility of data-driven techniques.

\begin{remark}\label{rem:params} 
{Note that \eqref{eq:bound-nnetm} is a worst-case upper bound for the consensus error under NN-ETM, by considering the maximum threshold $\sigma+\varepsilon$. However, if the consensus error is penalized during training, the NN learns to reduce the threshold such that the actual error is improved with respect to this bound.}
While providing freedom for the NN to optimize performance, NN-ETM relies on two user-defined constants $\sigma,\, \varepsilon$, which are necessary to formally guarantee stability and the exclusion of Zeno behavior. {We discuss their choice in Section \ref{sec:choice-params}.}
\end{remark}

\begin{remark}\label{rem:disturbances}
{The results from Proposition \ref{prop:iss} and Corollary \ref{cor:performance-nnetm} are general enough to take into account other external disturbances to the consensus protocol. Their effect can be represented by the constant $c$, or by additional $\gamma(\bullet)$ terms. As long as the protocol is ISpS with respect to them, the use of NN-ETM does not compromise its stability guarantees.}
\end{remark}

\subsection{Neural Network Architecture}\label{sec:architecture}

We set $\eta_i(t)$ as the output of a NN, defining a nonlinear map from the input information available at agent $i\in\mathcal{V}$ to the interval $[0,1]$. To design an appropriate architecture for the NN, we take into account the desired input and output information. 
First, the available information to be used as input at an agent $i$ is composed of the following:
\begin{itemize}
    \item \textbf{Local data:} Agent $i$'s local variable $\{\mf{z}_i(t') \ | \ t' \in [0, t]\}$, its sequence of events $\tau_0^i, \tau_1^i, \dots, \tau_t^i$ and the information transmitted at events, given by $\{ \mf{m}_i(\mf{z}_i(\tau_k^i)) \ | \ \tau_k^i < t \}$.
    \item \textbf{Data from neighbors:} The number of neighbors $| \mathcal{N}_i |$, their sequence of events $\tau_0^j, \tau_1^j, \dots, \tau_t^j$ and the information received when a neighbor triggers an event, $\{ \mf{m}_j(\mf{z}_j(\tau_k^j)) \ | \ j \in \mathcal{N}_i, \, \tau_k^j < t\}$.
\end{itemize}
These pieces of information are the only ones that each agent can use to maintain a distributed solution. Different combinations of these variables can be used as input, according to the desired behavior for the NN-ETM. {To maintain a general approach, we recommend choosing inputs such as the disagreement with neighbors, the event-triggered error and the time since last event, which can be computed from the data above. These inputs give information on how well an event-triggered consensus protocol is performing, regardless of its specific design.} Then, the input layer of the NN needs to have appropriate dimensions for the information to be used. 
As output, we require $\eta_i(t) \in [0,1]$. To constrain it to this interval, the output layer has only one neuron and an appropriate activation function, such as a sigmoid, which takes values in the desired range. 

{If the input and output requirements are met, under the proposed framework there is freedom of design for the internal details of the NN architecture. Among the simplest options is to use a multi-layer perceptron, which is composed of the input and output layers, along with one or more hidden layers and nonlinear activation functions. This architecture is known to be a universal function approximator when a sufficient number of hidden neurons is used. We use this option in our proof of concept in Section \ref{sec:experiments}, showing that it obtains a suitable adaptive behavior without requiring a computationally complex architecture. However, more sophisticated alternatives can be used, such as recurrent NNs, where previous outputs of the NN are also exploited.}

Note that such general structure of the NN-ETM subsumes different ETMs in the literature: for example, learning a fixed value of $\eta_i(t)$ regardless of the input would be equivalent to a fixed threshold approach, while the internal state held by recurrent NNs may emulate the dynamic variable in DETMs.

\subsection{Training the NN-ETM}\label{sec:training}
Now, we propose a training procedure to learn the weights inside the NN in our NN-ETM. We use the common approach of backpropagation, which computes the gradient of a cost function with respect to the NN's parameters, in order to optimize them using gradient descent. The goal of the training process is to find NN weights that minimize a cost function.

We explain our process for an average consensus algorithm, which we will use in our proof of concept in Section \ref{sec:experiments}, as an example. However, note that this process can be generalized to other consensus protocols, with minimal modifications. We consider a linear dynamic average consensus algorithm \cite{Kia2019}, adapted to event-triggered communication,
\begin{equation}\label{eq:lin-cons}
\begin{aligned}
    \dot{\mf{z}}_i(t) &= \dot{\mf{r}}_i(t) - \kappa \textstyle\sum_{j \in \mathcal{N}_i} ( \mf{z}_i(\tau_t^i) - \mf{z}_j(\tau_t^j) ) ,
\end{aligned}
\end{equation}
where $\mf{m}_i(\bullet) = (\bullet)$ and $\mf{r}_i(t) \in \mathbb{R}^n$ is a local reference signal for agent $i$. The goal for the agents is to reach agreement as $\mf{z}_i(t)=\bar{\mf{z}}(t) = (1/N)\sum_{i=1}^N \mf{z}_i(t)$. Such consensus trajectory will correspond to the average of the reference signals $(1/N)\sum_{i=1}^N \mf{r}_i(t)$.

As summarized in Algorithm \ref{alg:training}, the proposed NN-ETM is trained by running the event-triggered consensus algorithm over multiple pre-generated simulations of reference signals, computing a cost for the performance for each one, and performing the gradient backpropagation and weight update for the NN. {The simulation in step 8 of the training process is similar to the implementation shown in Algorithm \ref{alg:deployment}.} We repeat this process for several iterations (epochs), {monitoring the value of the cost function over the training epochs. We stop the training process when the cost converges to a steady value, which we use as a stopping criterion. In our case, we consider such convergence when the difference of the cost between epochs remains smaller than 0.001.} 

\begin{algorithm}
\begin{algorithmic}[1]
\State Generate batch of sequences for agents' local $\mf{r}_i(t)$ 
\State Set parameters $\sigma, \varepsilon$ for \eqref{eq:trigger-learning} 
\State Load NN model and initialize weights 
\State Set optimizer and training configuration 
\For {each training epoch} 
    \State cost $\leftarrow$ 0 
    \For{each sequence in the batch} %
        \State Simulate \eqref{eq:lin-cons} for $t\in[0,T]$ with current weights \phantom{forfor\ } (using fuzzy ETM from Section \ref{sec:fuzzy})
        \State Compute cost $\mathcal{J}$ \eqref{eq:cost-function} for the simulation 
        \State cost $\leftarrow$ cost + $\mathcal{J}$ 
    \EndFor
    \State Update network weights using backpropagation 
\EndFor
\caption{Training of NN-ETM}\label{alg:training}
\end{algorithmic}
\end{algorithm}

Despite the theoretical advantages of a general NN, appropriately training such models requires several heuristics and strategies often used in the literature. In the following, we detail some relevant considerations we used for the training process, which may be helpful for practitioners. Recall that the stability guarantees provided earlier are true regardless of this training procedure, which is only used to improve performance in terms of accuracy and network usage.

\subsubsection{Parameter Sharing and Scalability}\label{sec:parameter-sharing}
We consider that all agents use a copy of the same NN. For this reason, we use \emph{parameter sharing} \cite{Gupta2017} to train the network. That is, the NN for $\eta_i(t)$ shares the same weights for all agents $i \in \mathcal{V}$ during the training stage. The learning process is done using the experience of all agents involved. However, at the deployment stage, each agent can have a different behavior due to using its own local information as an input for the NN. {This perspective is advantageous since the NN can be trained using a small number of agents and tested for scalability with a larger number of agents when deployed, as we will illustrate in the simulation experiments from Section \ref{sec:experiments}.}

\subsubsection{Generation of Agents' Reference Signals {and Disturbances}}\label{sec:data-generation}

We generate a \textit{batch} of sequences of reference signals prior to the training process, which is of the form $\{\mf{r}_i(t)\}_{i=1}^N, t\in\{0,h,2h,\dots,T\}$. Here, $h>0$ is a discretization step for the simulation, and $T>0$ is the experiment length, as an integer multiple of $h$. Generating simulations of the consensus algorithm over the whole batch and back-propagating using the resulting cost constitutes one training epoch. {We consider here a consensus protocol without external disturbances, other than the time-varying reference signals preventing agreement. However, if the effect of other disturbances needs to be taken into account, batches of sequences for the disturbance's realization can also be generated similarly to the reference signals.}

\subsubsection{Cost Function}\label{sec:cost}

In event-triggered consensus problems, there is an evident trade-off between the communication rate of the agents and the overall performance of the consensus algorithm. Therefore, a balance needs to be found between the consensus error and communication load. 
For a given experiment of length $T$, let the mean square error $\mathcal{E}$ of the consensus variables $\mf{z}_i(t)$ with respect to the desired consensus trajectory $\bar{\mf{z}}(t)$ and the communication rate $\mathcal{C}$ be 
\begin{equation}
    \mathcal{E} = \frac{h}{N T} \sum_{i=1}^N \sum_{k=0}^{T/h} \| \bar{\mf{z}}(kh) - \mf{z}_i(kh) \|^2, \ \ \ \mathcal{C} = \frac{h}{N T} \sum_{i=1}^N e_i,
\end{equation}
where $e_i$ is the number of events triggered by agent $i$ during $t\in[0,T]$ and $h$ is the simulation step. Note that $\mathcal{C}$ is normalized between 0, being no communication, and 1, representing communication at each step of the simulation.
Moreover, some dynamic consensus algorithms have a non-zero steady-state error even when run at full communication (e.g., the one in \eqref{eq:lin-cons}). Therefore, we take into account the relative error of the event-triggered implementation with respect to the full communication case, since the error cannot be optimized beyond a certain point. Denoting as $\mathcal{E}_{\mathsf{fc}}$ the error for the full communication case and $\mathcal{E}_{\mathsf{ev}}$ the event-triggered case, the relative error is given by $\mathcal{E}_\mathsf{r} = (\mathcal{E}_{\mathsf{ev}} - \mathcal{E}_{\mathsf{fc}})/\mathcal{E}_{\mathsf{fc}}$.
Then, we design the cost function to be minimized by the NN as follows, describing the desired trade-off between error and communication,
\begin{equation}\label{eq:cost-function}
    \mathcal{J} = \mathcal{E}_\mathsf{r} + \lambda \mathcal{C},
\end{equation}
where $\lambda \geq 0$ is a user-defined parameter to assign relative weight to the error and communication costs. Increasing $\lambda$ will prioritize reducing the communication load at the cost of a higher consensus error, and vice versa.

\subsubsection{Fuzzy Event-Triggering Mechanism}
\label{sec:fuzzy}

Note that the training process relies on the computation of gradients, which requires the operations involved to be differentiable. The \emph{if-else} event-triggering decision, however, is non-differentiable: it results in two separate information updates in case of triggering an event or not doing so. To overcome this issue, we \emph{fuzzify} the ETM during training. Consider a sigmoid function
$$
\nu_i(t) =  \sigmoid(\alpha \ ( \| \mf{z}_i(t) - \mf{z}_i(\tau_k^i) \| -\delta(t))),
$$
where $\alpha \gg 1$ is a constant to make the sigmoid function steeper. If $\| \mf{z}_i(t) - \mf{z}_i(\tau_k^i) \| \geq \delta(t)$, an event is triggered with $\nu_i(t) \approx 1$, and vice versa. Then, the information update for an arbitrary variable $\mf{v}(t)$, that takes a value $\mf{v}_{\mathsf{event}}(t)$ if an event is triggered by agent $i$ at time $t$ and $\mf{v}_{\mathsf{no\ event}}(t)$ otherwise, can be approximated in a differentiable fashion as follows:
\begin{equation}
    \mf{v}(t) \leftarrow \nu_i(t) \mf{v}_{\mathsf{event}}(t) + (1-\nu_i(t)) \mf{v}_{\mathsf{no\ event}}(t).
\end{equation}
This approximation is only needed for the training process. At test time, the regular NN-ETM is used.

\subsubsection{Weight Initialization}

To facilitate convergence to an advantageous solution for the network's weights, we provide a suboptimal solution as the initial condition: a fixed event threshold. We achieve this by means of a pre-training process, in which the network learns to output $\eta_i(t) = 0.5$ at every step regardless of the input values. The pre-training also follows the process described in Algorithm \ref{alg:training}, but the cost function $\mathcal{J}$ is replaced with $\mathcal{J}'= \sum_{i=1}^N \sum_{k=0}^{T/h}|\eta_i(kh) - 0.5|$. We use the weights learned at this stage as initial weights in Algorithm \ref{alg:training}.

\subsection{{Implementation of NN-ETM for consensus}}\label{sec:implementation}
{Algorithm \ref{alg:deployment} summarizes the implementation of the trained NN-ETM for deployment.} In our implementation, we consider an Euler discretization of the protocol \eqref{eq:lin-cons} for step 7, but other suitable discretizations can be used for different protocols.

\begin{algorithm}
\begin{algorithmic}[1]
\State Set parameters $\sigma, \varepsilon$ for \eqref{eq:trigger-learning} 
\State Load trained NN model 
\State Initial event: $\mf{z}_i(\tau_t^i) \gets \mf{z}_i(0)$ 
\State Transmit $\mf{z}_i(0)$ to neighbors $j \in \mathcal{N}_i$ 
\For {all simulation steps $t \in \{0, h, 2h,\dots, T\}$} 
\For {all agents $i \in \mathcal{V}$} 
\State Run step of consensus protocol, e.g., \eqref{eq:lin-cons} 
\State Compute $\eta_i(t)$ in \eqref{eq:trigger-learning} by inference of the NN 
\State Check triggering condition \eqref{eq:send:on:delta}-\eqref{eq:trigger-learning}:
\If {event}
    \State $\mf{z}_i(\tau_t^i) \gets \mf{z}_i(t)$
    \State Transmit $\mf{z}_i(\tau_t^i)$ to neighbors $j \in \mathcal{N}_i$ 
\EndIf
\EndFor
\EndFor
\caption{{Implementation of NN-ETM for consensus}}\label{alg:deployment}
\end{algorithmic}
\end{algorithm}

\section{ISS Consensus under NN-ETM}\label{sec:iss-consensus}

In this section, we provide some examples of theoretical analysis for different consensus protocols. The goal for this section is twofold. First, we want to highlight that the proposed ISS theoretical framework to decouple the analysis of the consensus protocol from the choice of ETM can be used for different kinds of consensus protocols. We present two cases of such analysis, with different choices of $\mf{f}_i,\mf{m}_j$ for linear and nonlinear protocols in the literature, adapted to event-triggered communication. The second objective is to show that, with our proposed NN-ETM, we can explicitly write a bound for the consensus disagreement by using ISS analysis. The bound is dependent on the choice of parameters $\sigma, \varepsilon$ in \eqref{eq:trigger-learning}. Therefore, by this analysis, the parameters can be chosen such that a maximum admissible consensus error is guaranteed.

\subsection{Case I: ISpS Linear Consensus}

Consider the event-triggered linear dynamic consensus algorithm given in \eqref{eq:lin-cons}, with scalar values of $z_i(t), \, r_i(t)$ for simplicity, which can be equivalently written as
\begin{equation}\label{eq:lin-cons:iss}
\dot{z}_i(t) = \dot{r}_i(t) - \kappa \sum_{j \in \mathcal{N}_i} ( z_i(t)+u_i(t) - z_j(t) - u_j(t) ) ,
\end{equation}
with $u_i(t) = z_i(\tau_t^i)-z_i(t)$. Setting $\mf{r}(t) = [r_1(t), \dots, r_N(t)]^\top$, we write the protocol \eqref{eq:lin-cons:iss} in vector form, similarly to \eqref{eq:consensus-ev-vec}, as
\begin{equation}
\begin{aligned}
    \dot{\mf{z}}(t) &= \mf{f}(t, \mf{z}(t), \mf{u}(t)) = \dot{\mf{r}}(t) - \kappa \mf{Q}_\mathcal{G} \mf{z}(t) - \kappa \mf{Q}_{\mathcal{G}} \mf{u}(t).
\end{aligned}
\end{equation}
Then, recalling that $\mf{x}(t) = \mf{H}\mf{z}(t)$ and $\mf{H}\mf{Q}_\mathcal{G} = \mf{Q}_\mathcal{G}$, the dynamics for the disagreement, according to \eqref{eq:error-dyn}, are 
\begin{equation}
\label{eq:dis:linear}
\begin{aligned}
    \dot{\mf{x}}(t) &= \mf{h}(t, \mf{x}(t), \mf{u}(t)) = \mf{H} \dot{\mf{z}}(t) \\&= \mf{H} \dot{\mf{r}}(t) - \kappa \mf{Q}_\mathcal{G} \mf{x}(t) - \kappa \mf{Q}_{\mathcal{G}} \mf{u}(t).
\end{aligned}
\end{equation}

\begin{theorem}\label{th:linear}
Assume that $|\dot{r}_i(t)|\leq R, \, \forall t\geq 0$, for some $R\geq 0$. Then, the disagreement dynamics \eqref{eq:dis:linear} are ISpS. Moreover, under the NN-ETM \eqref{eq:send:on:delta}-\eqref{eq:trigger-learning}, the disagreement is bounded by
\begin{equation}\label{eq:linearcons-nnetm}
     \|\mf{x}(t)\| \leq \exp(-\kappa \lambda_2(\mathcal{G}) t) \|\mf{x}(0)\| + \frac{\lambda_{\max}(\mf{Q}_\mathcal{G})}{ \lambda_2(\mathcal{G})} (\sigma + \varepsilon) + \frac{\sqrt{N}R}{\kappa \lambda_2(\mathcal{G})}.
\end{equation}
\end{theorem}
\begin{proof}
It can be verified that the disagreement's dynamics have the following explicit solution:
\begin{equation}
\begin{aligned}
\mf{x}(t) &= \exp(-\kappa \mf{Q}_\mathcal{G} t)\mf{x}(0) \\&\quad + \textstyle\int_{0}^{t} \exp(-\kappa \mf{Q}_\mathcal{G} s)(\mf{H} \dot{\mf{r}}(t - s) - \kappa \mf{Q}_{\mathcal{G}} \mf{u}(t - s)) \nd s.
\end{aligned}
\end{equation}
Since $\mf{x}(t) = \mf{H}\mf{z}(t) = \mf{H}\mf{x}(t)$ due to $\mf{H}=\mf{H}^2$, then
\begin{equation}
\begin{aligned}
&\|\mf{x}(t)\| \leq \| \mf{H}\exp(-\kappa \mf{Q}_\mathcal{G} t)\mf{x}(0)\| \\&\quad + \textstyle\int_{0}^{t} \| \mf{H}\exp(-\kappa \mf{Q}_\mathcal{G} s)(\mf{H} \dot{\mf{r}}(t - s) - \kappa \mf{Q}_{\mathcal{G}} \mf{u}(t - s)) \| \nd s. 
\end{aligned}
\end{equation}
For the first term, we have that
$\| \mf{H} \exp(-\kappa \mf{ Q}_\mathcal{G} t)\mf{x}(0) \| \leq \lambda_{\max}(\mf{H} \exp(-\kappa \mf{Q}_\mathcal{G} t)) \|\mf{x}(0)\| = \exp(-\kappa \lambda_2(\mathcal{G}) t) \|\mf{x}(0)\| =: \beta(\|\mf{x}(0) \|, t),$
recalling that $\lambda_{\max}(\bullet)$ represents the largest eigenvalue of a matrix and $\lambda_2(\mathcal{G})$ denotes the algebraic connectivity of the graph.
For the second term, we have
\begin{equation}
\begin{aligned}
&\textstyle\int_{0}^{t} \| \mf{H}\exp(-\kappa \mf{Q}_\mathcal{G} s)(\mf{H} \dot{\mf{r}}(t - s) - \kappa \mf{Q}_{\mathcal{G}} \mf{u}(t - s)) \| \nd s \\
&\leq \sup_{s \in [0, t]} \| \dot{\mf{r}}(s) \| \textstyle\int_{0}^{t} \| \mf{H}\exp(-\kappa \mf{Q}_\mathcal{G} s) \mf{H}\| \nd s  \\ &\quad + \kappa \lambda_{\max}(\mf{Q}_\mathcal{G}) \sup_{s \in [0, t]} \|\mf{u}(s)\| \textstyle\int_{0}^{t} \| \mf{H}\exp(-\kappa \mf{Q}_\mathcal{G} s) \mf{H}\| \nd s \\
&\leq \frac{1-\exp(-\kappa \lambda_2(\mathcal{G})t)}{\kappa \lambda_2(\mathcal{G})} \big(\sup_{s \in [0, t]} \| \dot{\mf{r}}(s) \| \\ &\quad + \kappa \lambda_{\max}(\mf{Q}_\mathcal{G}) \sup_{s \in [0, t]} \|\mf{u}(s)\|\big) \\
&\leq \frac{1}{\kappa \lambda_2(\mathcal{G})} \left(\sup_{s \in [0, t]} \| \dot{\mf{r}}(s) \| + \kappa \lambda_{\max}(\mf{Q}_\mathcal{G}) \sup_{s \in [0, t]} \|\mf{u}(s)\|\right).
\end{aligned}
\end{equation}
Therefore, it follows that
$\gamma\textstyle\left(\sup_{s\in[0,t]} \| \mf{u}(s)\|\right) := \dfrac{\lambda_{\max}(\mf{Q}_\mathcal{G})}{ \lambda_2(\mathcal{G})}  (\sup_{s \in [0, t]} \|\mf{u}(s)\|).$
Note that $\| \dot{\mf{r}}(t) \| \leq \sqrt{N}R$ holds by assumption. Hence, \eqref{eq:iss} follows with
$c := \sqrt{N}R /(\kappa \lambda_2(\mathcal{G}))$, showing that the protocol is ISpS.
Finally, since we now have the functions $\beta, \, \gamma$ and the constant $c$, we can use the result in Corollary \ref{cor:performance-nnetm} to particularize the convergence result to our NN-ETM, yielding \eqref{eq:linearcons-nnetm}, which completes the proof.
\end{proof}

\subsection{Case II: ISS Nonlinear Consensus}
Consider the following nonlinear consensus protocol using $m$th-order sliding modes from \cite{noisy-edcho}. Here, we set $n=m+1$ for some arbitrary $m$ which depends on the application, and $\mf{z}_i(t)=[z_{i,0}(t),\dots,z_{i,m}(t)]^\top$, $\mf{m}_i(\mf{z}_i) = z_{i,0}$. While \cite{noisy-edcho} considers continuous-time communication, we adapt it to event-triggered communication as follows,
\begin{equation}\label{eq:noisy-edc}
\begin{aligned}
&\begin{array}{rl}
\dot{z}_{i,\mu}(t) =& z_{i, \mu+1}(t)-k_\mu \sum_{j \in \mathcal{N}_i} \sgn{z_{i,0}(\tau_t^i) - z_{j,0}(\tau_t^j)}{\frac{m-\mu}{m+1}},\\ & \text{for } 0 \leq \mu < m ,\\
\dot{z}_{i,m}(t) =&r_i^{(m+1)}(t)-k_m \sum_{j \in \mathcal{N}_i} \sgn{z_{i,0}(\tau_t^i) - z_{j,0}(\tau_t^j)}{0}, \\ 
\end{array}
\end{aligned}
\end{equation}
for appropriate design parameters $\{k_\mu\}_{\mu=0}^m $. Here, $r_i(t)$ is a reference signal for agent $i$. Moreover, $\lceil x\rfloor^\alpha:= |x|^\alpha\text{sign}(x)$ for $\alpha>0$ and $\lceil x\rfloor^0:={\text{sign}}(x)$.

Similarly as before, we write \eqref{eq:noisy-edc} in terms of a disturbed version of the ideal system with continuous time communication. Hence, write the system in vector form as
\begin{equation}
\begin{aligned}
&\begin{array}{rl}
\dot{\mf{z}}_\mu(t) =& \mf{z}_{\mu+1}(t) - k_\mu \mf{D}_\mathcal{G}\sgn{\mf{D}_\mathcal{G}^\top(\mf{z}_0(t) + \mf{u}(t))}{\frac{m-\mu}{m+1}}, \\
\dot{\mf{z}}_m(t) =& \mf{r}^{(m+1)}(t) - k_m \mf{D}_\mathcal{G}\sgn{\mf{D}_\mathcal{G}^\top(\mf{z}_0(t) + \mf{u}(t))}{0}, 
\end{array}
\end{aligned}
\end{equation}
with $\mf{z}_\mu(t) = [z_{1, \mu}(t), \dots,z_{N, \mu}(t)]^\top$ and $\mf{u}(t)$ with $u_i(t) = z_{i,0}(\tau_t^i)-z_{i,0}(t)$. As in \cite{noisy-edcho}, since $\mf{D}_{\mathcal{G}}^\top \mf{H} = \mf{D}_{\mathcal{G}}^\top$, it can be verified that the disagreement dynamics $\mf{x}_\mu(t) = \mf{H}\mf{z}_\mu(t)$ satisfy the differential inclusion
\begin{equation}\label{eq:edcho-disagr}
\begin{aligned}
&\begin{array}{rl}
\dot{\mf{x}}_\mu(t) =& \mf{x}_{\mu+1}(t) - k_\mu \mf{D}_\mathcal{G}\sgn{\mf{D}_\mathcal{G}^\top(\mf{x}_0(t) + \mf{u}(t))}{\frac{m-\mu}{m+1}},  \\
\dot{\mf{x}}_m(t) \in& [-L, L]^{N} - k_m \mf{D}_\mathcal{G}\sgn{\mf{D}_\mathcal{G}^\top(\mf{x}_0(t) + \mf{u}(t))}{0}.
\end{array}
\end{aligned}
\end{equation}
\begin{theorem}\label{th:edcho}
Assume that there exists $R\geq 0$ such that $|{r}^{(m+1)}_i(t)|\leq R, \forall t\geq 0$. Set the parameters $\{k_\mu\}_{\mu=0}^m $ as in \cite{noisy-edcho}. Then, the disagreement dynamics in \eqref{eq:edcho-disagr} are ISS. Moreover, under the NN-ETM \eqref{eq:send:on:delta}-\eqref{eq:trigger-learning}, the disagreement is bounded by
\begin{equation}\label{eq:edc-nnetm}
\begin{aligned}
     \|\mf{x}(t)\| &\leq B(\|\mf{x}(0)\|)\exp(\mathcal{T}(\|\mf{x}(0)\|)-t) \\ &\quad + \sqrt{m+1}\max_{0\leq\mu\leq m}c_\mu (\sigma + \varepsilon)^{\frac{m-\mu+1}{m+1}},
\end{aligned}
\end{equation}
for some constants $c_\mu \geq 0$, for $\mu=0, \dots, m$, and setting
\begin{equation}\label{eq:b-t-functions}
\begin{aligned}
B(x) &:= \sup\{\|\mf{x}(t;\mf{x}(0))\|:\|\mf{x}(0)\|=x, t\in[0,T(\mf{x}(0))]\},\\
\mathcal{T}(x) &:= \sup\{T(\mf{x}(0)) : \|\mf{x}(0)\|=x\}.
\end{aligned}
\end{equation}
\end{theorem}
\begin{proof}
Given an arbitrary bound $U\geq 0$ for $\|\mf{u}(t)\|$, the results from \cite[Lemma 9]{noisy-edcho} imply that the solutions of \eqref{eq:edcho-disagr} converge to the regions $\| \mf{x}_\mu(t) \| \leq c_{\mu} U^{\frac{m-\mu+1}{m+1}}$ with constants $c_{\mu}\geq 0$, for $t\geq T(\mf{x}(0))$ and some finite time $T(\mf{x}(0))\geq 0$. Therefore,  
  $  \|\mf{x}(t)\| \leq \textstyle\sqrt{\sum_{\mu=0}^m c_\mu^2U^{2\frac{m-\mu+1}{m+1}}} \leq  \sqrt{m+1}\max_{0\leq\mu\leq m}c_\mu U^{\frac{m-\mu+1}{m+1}} =: \gamma(U) $
for all $t\geq T(\mf{x}(0))$. Denote with $\mf{x}(t;\mf{x}(0))$ the solution of \eqref{eq:edcho-disagr} given an initial condition $\mf{x}(0)$. Then, set the functions $B(x), \mathcal{T}(x)$ as in \eqref{eq:b-t-functions}, 
so that
$
\|\mf{x}(t)\|\leq B(\|\mf{x}(0)\|)\exp(\mathcal{T}(\|\mf{x}(0)\|)-t)=:\beta(\|\mf{x}(0)\|,t)$
for all $t\in[0,T(\mf{x}(0))]$. Hence, the inequality \eqref{eq:iss} follows for this choice of $\gamma,\beta$ for all $t\geq 0$, with $c:=0$, which shows that the consensus protocol is ISS. 
Recalling Corollary \ref{cor:performance-nnetm}, the result in \eqref{eq:edc-nnetm} follows, completing the proof.
\end{proof}

\begin{remark}\label{rem:iss-examples}
Besides complying with the ISpS property required in Theorem \ref{th:iss}, Cases I and II also ensure item \eqref{it:existence} in Theorem \ref{th:iss} as per the analysis in \cite{Kia2019} and \cite{noisy-edcho} respectively, for the full communication case. The same analysis in such works ensures item (\ref{it:boundedconsensus}) of Proposition \ref{prop:zeno} is satisfied by the consensus trajectory under appropriate initialization. Thus, both protocols under our NN-ETM comply with all the design criteria in Section \ref{sec:design:criteria}, {ensuring an upper error bound as \eqref{eq:bound-nnetm}.} 
\end{remark}

\subsection{{Choice of parameters in NN-ETM}}\label{sec:choice-params}

From Theorems \ref{th:linear} and \ref{th:edcho}, we have explicit worst-case bounds for the consensus disagreement, which depend on the design parameters of the consensus protocol, the properties of the graph, the bounds of the reference signals, and the choice of $\sigma, \varepsilon$ in the NN-ETM.
{Note that only the term $\gamma'(\sigma+\varepsilon)$ in the bound \eqref{eq:bound-nnetm} depends on the choice of parameters of NN-ETM, while the error terms $\beta(\|\mf{x}(0)\|, t)$ and $c$ are inherent to the particular consensus protocol. The steady-state consensus error bound is given by $\gamma'(\sigma+\varepsilon)+c$, since $\beta(\|\mf{x}(0)\|, t)$ is a $\mathcal{KL}$ function that vanishes over time. }
{Thus, in order to choose these parameters to guarantee a certain consensus error, the following procedure can be followed. First, decide the admissible worst-case bound $b$ for the consensus error in the given application scenario. Next, find appropriate functions $\beta,\, \gamma',\, c$ for the consensus protocol under consideration, following similar analysis to the proofs of Theorems \ref{th:linear} and \ref{th:edcho}. Set $\varepsilon$ arbitrarily small, recalling that we merely require $\varepsilon > 0$ to guarantee Zeno-freeness. Finally, comparing the admissible error $b$ to the bound \eqref{eq:bound-nnetm} (or the steady-state bound $\gamma'(\sigma+\varepsilon)+c$, if the transient is not relevant), obtain the value of $\sigma$.}
Increasing the values of $\varepsilon,\, \sigma$ enlarges the minimum inter-event time and allows further reductions in communication but also increases the admissible upper-bound for the consensus error, recalling that \eqref{eq:bound-nnetm} is a worst-case bound.

\section{Discussion}\label{sec:discussion}

\textit{Generalizability to different consensus protocols:} While we have developed and applied the design criteria in Theorem \ref{th:iss} to design a safe NN-ETM, this framework applies to a wide range of works in the literature. With regard to producing a bounded disturbance $\mf{u}(t)$, designs similar to the ETM in Proposition \ref{prop:iss} can provide this assurance, both for fixed-threshold \cite{Xing2020,event-edc}, adaptive \cite{Meng2015} or dynamic \cite{George2018} approaches. In terms of the ISpS requirement for the consensus protocol, we have already shown detailed examples of fulfillment in the cases from Section \ref{sec:iss-consensus}. Note that the ISpS result is achieved by other consensus protocols in the literature, under different ETMs, e.g., \cite{Liu2023,Qian2023,Xing2020,Kia2015}. Even the cases of \cite{George2018,Meng2015}, which achieve asymptotic convergence to the desired value with their proposed ETMs, can be analyzed under this framework, noting that their event thresholds vanish over time and thus the perturbation caused by it does as well, recalling Proposition \ref{prop:iss}. 
These examples serve not only to validate our proposed design criteria but also to highlight that they can be used to check that different combinations of ETMs and consensus protocols, not necessarily designed ad-hoc, can be connected while guaranteeing a bounded consensus error. Decoupling the conditions for each block facilitates the analysis of different combinations by checking their relevant requirements.
{Moreover, the generality of our framework enables the use of our proposed NN-ETM in different contexts and applications where consensus algorithms are leveraged, such as distributed state estimation \cite{ecc24}, formation control of multi-agent systems \cite{Ren2006}, and other consensus-based protocols to recover signals of interest in a distributed fashion, as in the distributed differentiation problem from \cite{noisy-edcho}.}

{\textit{Comparison to existing ETMs:} For any suitable consensus protocol, our NN-ETM aims to provide a general and adaptive solution. In contrast to fixed-threshold ETMs, the variable threshold in NN-ETM adapts the communication rate according to the performance of the consensus protocol, taking into account aspects such as the disagreement with respect to neighboring nodes. According to Corollary \ref{cor:performance-nnetm} and Remark \ref{rem:params}, the error bound for a consensus protocol under NN-ETM cannot be worse than under a fixed-threshold ETM with threshold value $\sigma+\varepsilon$, regardless of the NN architecture and for any graph size and topology.} With respect to DETMs, our approach is more general, as it can be applied to many classes of consensus protocols, whereas DETMs are tailored to a particular one. As we will show in Section \ref{sec:experiments}, existing DETMs focus on achieving agreement asymptotically, at the cost of having a potentially heavy communication load. In contrast, NN-ETM retains an adaptive nature and allows a more significant communication reduction. Finally, with respect to other learning-based approaches, we introduce the NN safely, providing strict bounds for the resulting consensus error. 

{\textit{Training process and complexity:} {Since we use the NN to fit a function for the threshold $\delta(t)$ according to the evolution of variables in the event-triggered consensus protocol, training over simulation is sufficient, in contrast to other complex tasks such as image detection that may require real data. A suitable function can be learned even with small NN architectures and batches of simulation data, which is beneficial to reduce the computational complexity of the training and inference processes.} 
In addition, other measures have been adopted to reduce the computational load, such as the use of the parameter sharing strategy for training described in Section \ref{sec:parameter-sharing}.
{Our training process requires the simulation of a given consensus protocol and a batch of reference signals. With appropriate choices of input variables to NN-ETM, the generality of the result is not necessarily compromised. As noted in Section \ref{sec:architecture}, choosing input variables such as the disagreement, event-triggered error and time since last event provides information regarding the performance of the consensus protocol, independently of its particular design, and of the network size. Different classes of reference signals can be generated, and the effect of other factors such as disturbances can be taken into account in the simulations. Our experiments from Section \ref{sec:experiments} suggest that consistent behavior is achieved when NN-ETM is tested over different network sizes and topologies, and even alongside a different consensus protocol than during training.} }

\section{Experiments}\label{sec:experiments}

We train and test our NN-ETM in simulation examples, according to the process detailed in Section \ref{sec:training}.\footnote{Our code is available at: \url{https://github.com/ireneperezsalesa/NN-ETM/}} 
In the following, we provide the details of the setup, training process, and results, {including comparisons with traditional methods and testing over different network topologies and sizes.}

\subsection{Setup}\label{sec:setup}
For our proof of concept, we consider an average consensus problem, where agents track the average of their time-varying local reference signals. For simplicity, we consider $z_i(t),\, r_i(t) \in \mathbb{R}$, and we generate sinusoidal reference signals as $r_i(t) = a_i \sin(\omega_i t + \phi_i)$, where  $a_i, \, \omega_i, \, \phi_i$ are randomly chosen in the ranges $[0, 50]$, $[0, 10]$ and $[0, \pi]$, respectively, so that agents have different signals. {We adopt such signals for the sake of generality, noting that reference signals for each agent may be generated from different sources in practical applications, e.g., by physical dynamics in formation control problems where the reference signal may correspond to an agent's position, or by numerical protocols run by the agents in distributed estimation or optimization algorithms.}
We consider the linear consensus protocol from \eqref{eq:lin-cons}, {as well as the nonlinear consensus protocol and DETM proposed in \cite{George2018}}. Both protocols fulfill the ISpS requirement for the consensus error and guarantee a bounded error when connected to our NN-ETM, as we have already formally shown in Theorem \ref{th:linear} for the linear case. For the linear protocol, we set $\kappa=50$. {For the DETM, we set the parameters similarly as in \cite{George2018}, with $\gamma =1,\, \alpha_i=3,\, \beta_i=100,\, \delta_i=1.5,\, \theta_i=0.9,\, \kappa_j(0)=10$, and $\eta'_i(0)=5$ for their dynamic variable. Note that a significant number of constants needs to be tuned by hand, affecting the resulting performance and communication savings.} In contrast, only two parameters need to be chosen in NN-ETM \eqref{eq:trigger-learning} according to our earlier discussions, which we set to $\sigma = 1, \, \varepsilon = 0.001$ here. {To verify the scalability of the proposal, we evaluate its performance over several network topologies, with $N=3$ agents (used also for training), $N=8$ and $N=10$ agents. The graph topologies are shown in Figure \ref{fig:graphs}. Furthermore, we will consider experiments with a time-varying network, switching between the three topologies, to illustrate the adaptability of NN-ETM for relevant practical cases, such as when agents spontaneously connect and disconnect in open networks.} We simulate all approaches through an Euler discretization with step $h=10^{-3}$, and we initialize the simulations to unfavorable values to verify the convergence of the consensus algorithms.
\begin{figure}
    \centering
    \includegraphics[width=\linewidth]{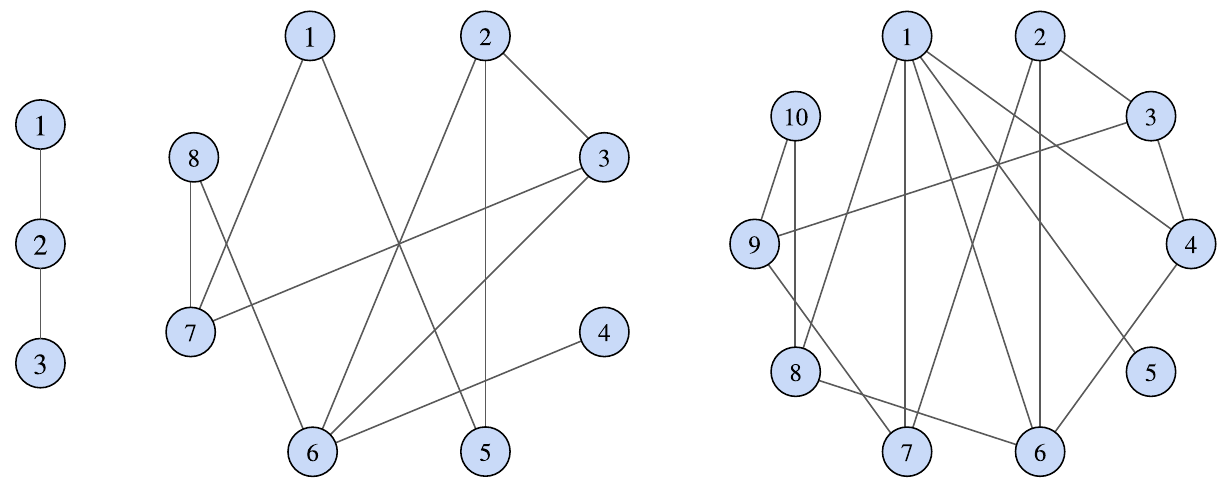}
    \caption{{Communication graphs used in the experiments.}}
    \label{fig:graphs}
\end{figure}

\subsection{NN Architecture and Training Details}\label{sec:experiments-training-dets}
{Recalling the general considerations from Sections \ref{sec:architecture} and \ref{sec:training}, we detail here the particular choices adopted for the following experiments.} We have chosen a multi-layer perceptron architecture for the NN in the NN-ETM \eqref{eq:trigger-learning}, in order to produce an adaptive threshold $\eta_i(t)$ as a function of the inputs. 
For the sake of comparison with existing DETMs in the literature, as inputs to an agent's NN, we have chosen the event-triggered error $|u_i(t)| = |z_i(t) - z_i(\tau_t^i)|$, the disagreement term with respect to its neighbors, $|d_i(t)| = |\sum_{j\in\mathcal{N}_i} ({z}_i(\tau_t^i) - {z}_j(\tau_t^j))|$, and the time since the last event triggered by the agent, $\Delta_i(t) = t - \tau_t^i$, so that $\eta_i(t)$ adapts accordingly to these values. {To show that an adaptive threshold can be learned even with a small NN, such that computational complexity is not significantly increased with respect to other ETMs, we have chosen a simple architecture with only one hidden layer, but other architectures can be used, as discussed in Section \ref{sec:architecture}.} The input layer in this example has dimension 3, corresponding to the dimension of the input, given by $[|u_i(t)|, |d_i(t)|, \Delta_i(t)]$. The hidden layer has 8 neurons, and the output layer has dimension 1 to produce the scalar value $\eta_i(t)$. We have used ReLU activation functions for the non-output layers and a sigmoid activation at the output, ensuring $\eta_i(t)\in[0,1]$.
We use PyTorch's \emph{autograd} \cite{paszke2019pytorch}, along with the Adam optimizer with learning rate $10^{-3}$. We train the NN-ETM using the network of $N=3$ agents. For the training stage, we have generated a batch of $5$ sequences of length $T=10$, with step $h=10^{-3}$. The training process aims to minimize the cost function \eqref{eq:cost-function}, recalling that $\lambda$ is a user-defined parameter. We have repeated the training process for different values of $\lambda$, to validate the fact that the NN learns to adapt the behavior of the setup, balancing the error vs. communication trade-off.

\subsection{Simulation Results}\label{sec:exp-results}
{In the following experiments, we aim to validate the following aspects: (a) the proposed training process produces a suitable threshold function $\eta_i(t)$ according to how much communication is penalized by the weight $\lambda$ in the cost function, (b) the triggering policy is scalable to a larger agent network, (c) an adaptive threshold that evolves with the input signals is learned, which reduces communication with respect to fixed-threshold and dynamic ETMs, and (d) the learned policy can generalize to consensus protocols different from the one used for training, in contrast to hand-crafted DETMs.}

{We first test the performance of different setups over the same reference signals and initial conditions, mainly over the small network with $N=3$ for visualization purposes. Figure \ref{fig:lin-nnetm-lambdas} shows a simulation of the event-triggered consensus protocol and the evolution of $\eta_i(t)$ for each agent, for different values of $\lambda$ in the cost function during training. We change the reference signals at $t=0.3$ in this experiment, to validate the adaptiveness of the threshold in response to sudden changes. During the transient periods where disagreement is high and signals change rapidly, the value of $\eta_i(t)$ increases to reduce the number of triggered events. This behavior is consistent with DETMs in the literature. Once the disagreement between agents decreases, a steady state is reached where the value of the event threshold depends on how much communication was penalized during training: if a high $\lambda$ was used in the cost function \eqref{eq:cost-function}, the threshold takes larger values to further reduce communication by allowing a non-zero steady-state consensus error. If $\lambda$ is small, the threshold takes smaller values to improve the consensus error in steady state at the cost of additional communication. Thus, the training process is successful in learning a suitable $\eta_i(t)$ for each case. In the case of small $\lambda$, the steady-state threshold looks almost constant, but for higher $\lambda$ where larger values of $u_i(t), d_i(t)$ are allowed, a more significant variability of the threshold is seen.}

\begin{figure*} 
    \centering
  \subfloat[{$\lambda = 10^{-3}$}\label{fig:lin-nnetm-l0001}]{%
       \includegraphics[width=0.3\linewidth]{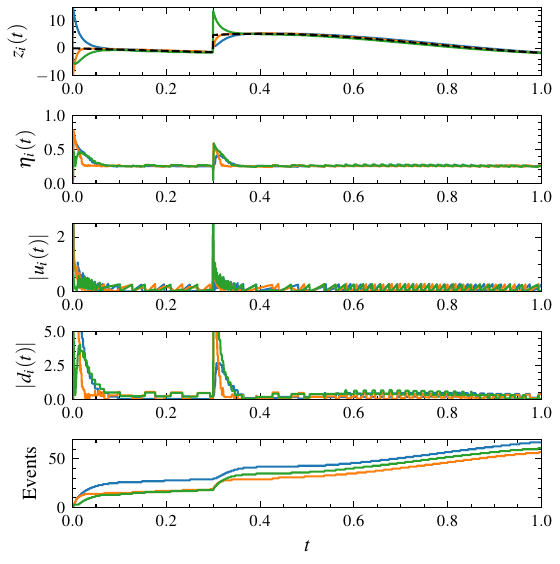}}
    \hfill
  \subfloat[{$\lambda = 10^{-2}$}\label{fig:lin-nnetm-l01}]{%
        \includegraphics[width=0.3\linewidth]{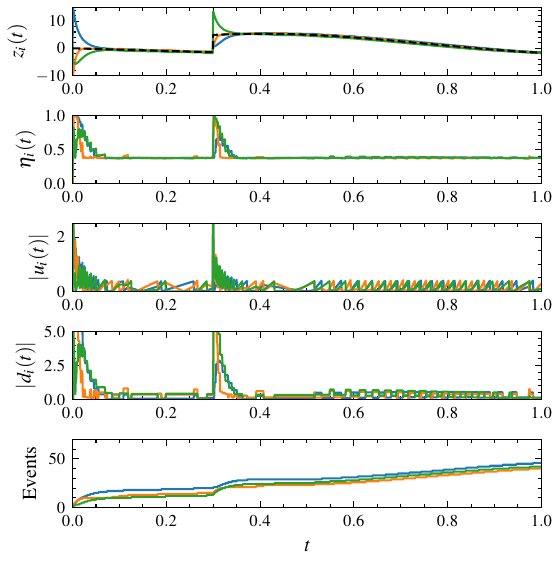}}
    \hfill
  \subfloat[{$\lambda = 1$}\label{fig:lin-nnetm-l1}]{%
        \includegraphics[width=0.3\linewidth]{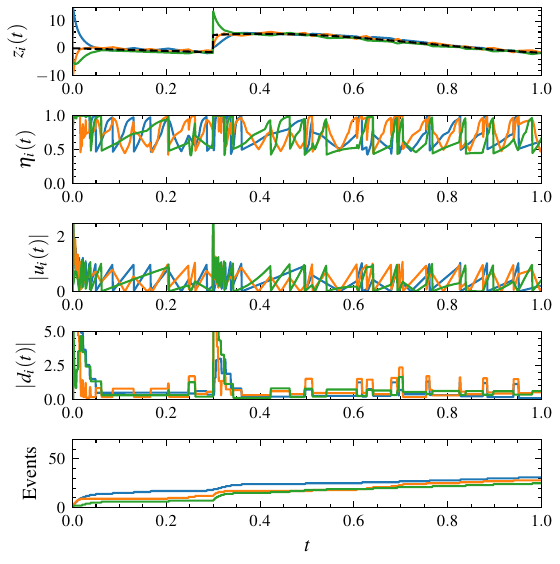}}
    \hfill
  \caption{Linear consensus under NN-ETM for $N=3$ agents, with different values of $\lambda$ in the training cost function. Increasing $\lambda$ leads to higher values of $\eta_i(t)$ in the event threshold, further reducing communication by allowing higher consensus errors.} 
  \label{fig:lin-nnetm-lambdas} 
\end{figure*}

{Note that the value of $\eta_i(t)$ is different for each agent, even though they all use the same weights in the NN. That is, while all agents share the same triggering policy (with the trained NN-ETM), the resulting $\eta_i(t)$ depends on the local observations. Figure \ref{fig:lin-nnetm-l0001-N10} shows an example with a larger network of $N=10$ agents, showing that the same policy holds for larger graphs, where each agent triggers according to its local signals. The behavior of $\eta_i(t)$ remains similar to the case with $N=3$, validating the scalability of the proposal.}

\begin{figure}
    \centering
    \includegraphics[width=0.6\linewidth]{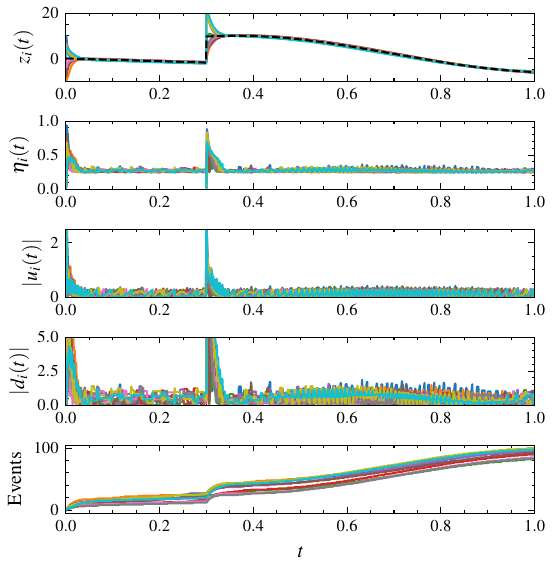}
    \caption{Linear consensus under NN-ETM ($\lambda=10^{-3}$) for the network with $N=10$ agents. The triggering policy remains similar to the case with $N=3$ agents, since each agent uses a copy of the same trained NN.}
    \label{fig:lin-nnetm-l0001-N10}
\end{figure}

{{To compare with the traditional fixed-threshold ETM, we have set a constant $\eta_i = 0.4$, which is a similar value to the steady-state $\eta_i(t)$ produced by NN-ETM with $\lambda=0.01$ in the simulation from Figure \ref{fig:lin-nnetm-lambdas}. The main difference between both approaches is seen in the transient periods, as shown in Figure \ref{fig:lin-closeups}. Thanks to the adaptive threshold, NN-ETM reduces the number of events} (by 27.8\% in this example) while maintaining comparable performance. This showcases the ability of NN-ETM to adapt the threshold value according to the time-varying conditions of the system, achieving similar performance while reducing unnecessary communication.}

\begin{figure} 
    \centering
  \subfloat[Fixed-threshold ETM\label{fig:lin-sod-closeup}]{%
       \includegraphics[width=0.6\linewidth]{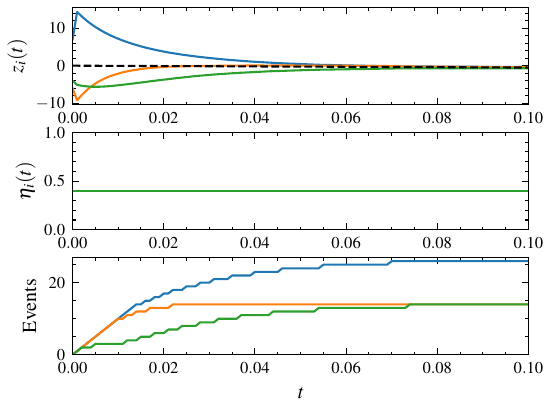}}
    \\
  \subfloat[NN-ETM\label{fig:lin-nnetm-closeup}]{%
        \includegraphics[width=0.6\linewidth]{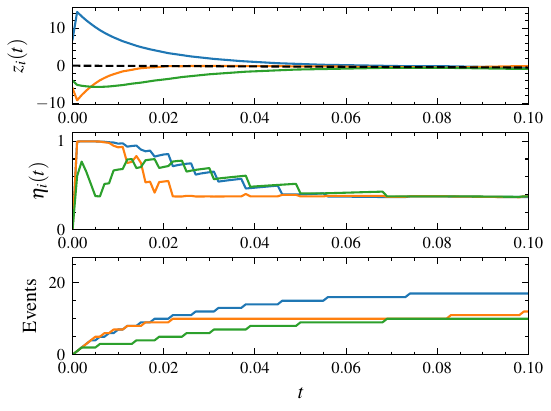}}
    \hfill
  \caption{Close-up of transient period with a fixed-threshold ETM ($\eta_i(t) = 0.4, \forall t$) and NN-ETM ($\lambda=0.01)$. The adaptive threshold in NN-ETM reduces the number of events compared to the fixed-threshold ETM, producing 17, 12 and 10 events for each agent instead of 26, 14 and 14 {(reduction of 27.8\%)}.}
  \label{fig:lin-closeups} 
\end{figure}

{Next, we compare our NN-ETM to the hand-crafted DETM from \cite{George2018}. We recall the difference in design philosophy of our NN-ETM and traditional DETMs: DETMs are tailored to a particular consensus protocol and do not necessarily generalize to other protocols, potentially even causing instability or Zeno behavior, while we have designed NN-ETM to provide such guarantees for any consensus protocol fulfilling Theorem \ref{th:iss}. Since the DETM may not generalize well, we compare both triggers over the nonlinear consensus protocol proposed in \cite{George2018}, such that the DETM's guarantees hold. The results are shown in Figure \ref{fig:george}. For the DETM, we show its dynamic variable $\eta_i'(t)$, although it is not exactly equivalent to our $\eta_i(t)$ in NN-ETM. Note that $\eta_i'(t)$ vanishes towards zero over time, which allows asymptotic convergence to the consensus value. However, events are triggered more frequently as $\eta'_i(t)$ vanishes, eventually triggering events at all simulation steps. Due to the nature of dynamic consensus, where a time-varying average is tracked, the nodes need to maintain communication to ensure that the signal is tracked correctly. Thus, for tasks with a longer temporal horizon, the reduction in communication from this method may be negligible. Moreover, the choice of initial $\eta_i'(0)$ has a significant effect on how much communication is saved. In contrast, NN-ETM retains adaptive behavior during the transient periods, additionally reducing the communication rate in steady state by allowing a small error around the desired agreement value.} {In this case, communication is reduced by 93.4\%.} {Note that we have not re-trained the NN for this particular consensus protocol, which could further improve performance, in order to evaluate how the policy learned over a simple consensus protocol can generalize to more complex ones.}

\begin{figure} 
    \centering
  \subfloat[DETM\label{fig:george-detm}]{%
       \includegraphics[width=0.6\linewidth]{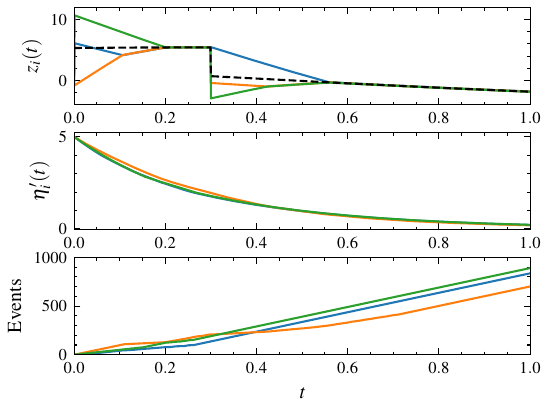}}
    \\
  \subfloat[NN-ETM\label{fig:george-nnetm}]{%
        \includegraphics[width=0.6\linewidth]{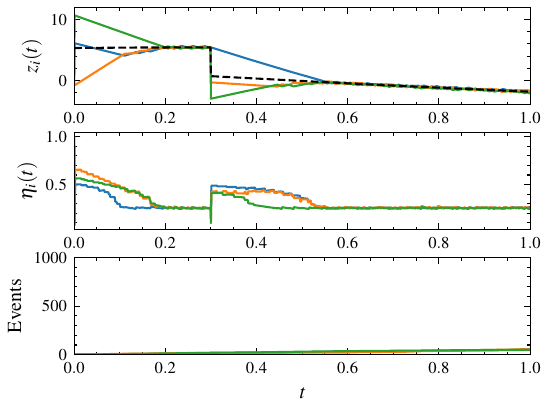}}
    \hfill
  \caption{Nonlinear consensus from \cite{George2018} under their DETM and NN-ETM ($\lambda=10^{-3}$). While the DETM allows asymptotic convergence, it only reduces communication in the initial transient. In contrast, NN-ETM can reduce events both in the transient and steady states, {producing a 93.4\% communication reduction in this example, with similar performance.}}
  \label{fig:george} 
\end{figure}

{So far, we have considered static graph topologies. However, a relevant issue in practical applications, such as mobile robotics, is the connection and disconnection of agents due to their movement, which causes changes in the graph topology and introduces additional disagreement due to poor initialization of the newly added agents. In the following experiment, we compare the performance of NN-ETM against the fixed-threshold and DETM approaches in a time-varying setting. We consider the consensus protocol from \cite{George2018}, since it is robust to such phenomena. Starting with $N=3$ connected agents as in the simplest graph (recalling that this is the topology that has been used for training), additional agents connect at $t=1$, switching to the $N=8$ agent topology from Fig. \ref{fig:graphs}, and two more agents connect at $t=3$ to form the third graph with $N=10$ agents. Since each agent is unaware of suitable initial values when it connects to the existing network, it causes a transient until agreement is recovered again. We consider that agents trigger communication events once they are connected to the network, and we fix their estimates $\hat{z}_i(t)$ to zero or their previous available value while disconnected. Fig. \ref{fig:switch} shows the results obtained for each method. At each change in the network topology, NN-ETM adapts the threshold $\eta_i(t)$ of all connected agents according to the disagreement perceived by each. This increase leads to a reduction in communication during transients. For the fixed-threshold approach, more events are triggered during transients, since the threshold does not adapt to account for the sudden additional disagreement, while steady-state performance remains similar to NN-ETM for a similar threshold value. Finally, for the DETM, newly connected agents start out with the initial value $\eta_i'(0)$ for their dynamic variable, reducing initial communication for the new agents. However, the threshold of the agents that were already connected to the network does not react to the increased disagreement when a new neighbor is connected, in contrast to the case for NN-ETM. In steady state, any communication savings from the DETM vanish as the event threshold becomes zero. Therefore, NN-ETM preserves the steady-state communication savings from the fixed-threshold ETM, while further reducing the number of events during transients thanks to the responsiveness of the event threshold towards changes in the agent's perceived disagreement when conditions are modified in their neighborhood.}

\begin{figure*} 
    \centering
  \subfloat[{NN-ETM}\label{fig:switch-nnetm}]{%
       \includegraphics[width=0.3\linewidth]{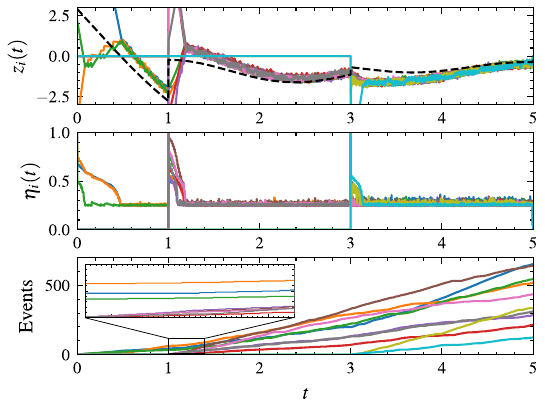}}
    \hfill
  \subfloat[{Fixed-threshold ETM}\label{fig:switch-sod}]{%
        \includegraphics[width=0.3\linewidth]{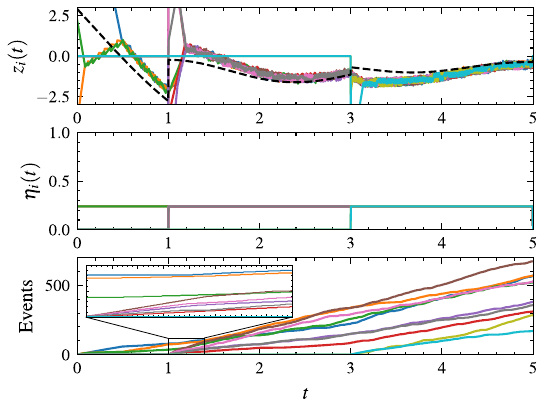}}
    \hfill
  \subfloat[{DETM}\label{fig:switch-detm}]{%
        \includegraphics[width=0.3\linewidth]{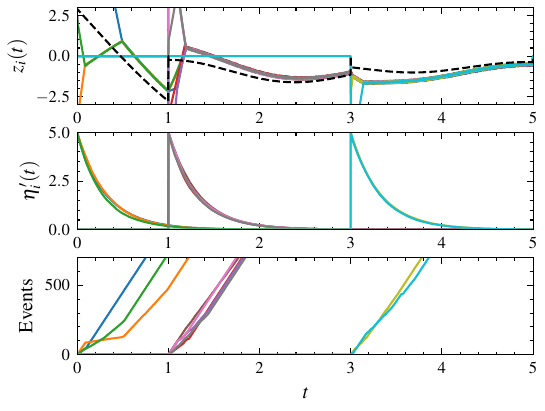}}
    \hfill
  \caption{{Results for the consensus protocol from \cite{George2018} under different ETMs, for a time-varying network topology where additional agents are connected at $t=1$ and $t=3$, going from $N=3$ agents to $N=8$ and $N=10$. NN-ETM's responsive event threshold accounts for the additional disagreement introduced by new agents to reduce communication during transients, while preserving the steady-state communication savings from the fixed-threshold case, while the DETM only reduces communication during the transient for the newly connected agents.}}
  \label{fig:switch} 
\end{figure*}

{Finally, to validate the performance of the proposed NN-ETM over different protocols, network topologies and sizes, we include a summary of results for the experiments conducted over batches of 20 sequences of length $T=5$, different than those used for training. All setups in the same scenario (i.e., the same consensus protocol and number of agents) use the same batch of reference signals and initial conditions to ensure a fair comparison. The averages of mean square errors and triggered events per node over the 20 simulations are provided in Table \ref{tab:sims}. We also include the values for the full communication case as a baseline, since the performance varies with the type of consensus protocol and number of agents, even with full communication. From these results, it is demonstrated that NN-ETM achieves a better trade-off than the fixed-threshold ETM (comparing it with the values for NN-ETM with $\lambda=10^{-3}$), simultaneously reducing the MSE and number of events with respect to this baseline, thanks to the adaptable threshold during transients. Compared to the DETM, the communication savings are very significant, as discussed earlier. These aspects can be seen across all scenarios, with $N=3$, $N=10$, and the switching topology with $N \in \{3, 8, 10\}$, for both consensus protocols, validating the scalability of the proposal and its applicability to different consensus protocols and network topologies.}

\subsection{{Application example}}

{A prominent example of the use of consensus in a systems engineering application is to achieve collaborative target tracking in a team of robots. We consider an example where a network of $N=10$ follower agents connected in a ring topology needs to track the trajectories of 3 additional leader agents for use in formation control. Each robot is equipped with sensors and onboard processing units and can communicate in a wireless fashion with neighboring agents, although reducing communication is desired to prolong their battery life. Agents model the dynamics of each leader $\ell$ as
\begin{equation}
    \dot{\mf{x}}_\ell(t) = \begin{bmatrix}
        0 & 0 & 1 & 0 \\ 0 & 0 & 0 & 1 \\ 0 & 0 & 0 & 0 \\ 0 & 0 & 0 & 0 \\
    \end{bmatrix} \mf{x}_\ell(t) + \begin{bmatrix}
        0&0\\0&0\\1&0\\0&1\\
    \end{bmatrix}\mf{w}_\ell(t),
\end{equation}
with $\mf{x}_\ell(t) = [x_\ell(t),y_\ell(t),v_{\ell,x}(t), v_{\ell,y}(t)]$ representing the 2-D Cartesian position and velocity coordinates and $\mf{w}_\ell(t)$ modeling the unknown acceleration as $\mf{w}_\ell(t) \sim \mathcal{N}(\mf{0}, \mf{I})$.}
{We assume that each follower can only obtain a measurement of its closest leader agent, with agents $i \in \{1,2,3\}$ measuring the first leader's position, $i \in \{4,5,6,7\}$ measuring the second, and $i \in \{8,9,10\}$ observing the third, but they all need to estimate the trajectories of all leaders, given by the full state $\mf{x}(t) = [\mf{x}_1(t)^\top, \mf{x}_2(t)^\top, \mf{x}_3(t)^\top]$. The noisy measurements are given by $\mf{y}_i(t)=\mf{C}_i \mf{x}(t) + \mf{v}_i(t)$, with appropriate matrices $\mf{C}_i$ and $\mf{v}_i(t) \sim \mathcal{N}(\mf{0}, 0.01\mf{I})$.}
{To track the leaders' positions in a distributed fashion from partial and noisy local information, we implement the consensus-based filter from \cite{ecc24}. This filter allows us to compute distributed estimates of $\mf{x}(t)$ by exchanging state estimate information only at events. In the absence of noise, the filter fulfills the conditions from Theorem \ref{th:iss}. Under stochastic noise, the filter has been proven in \cite{ecc24} to remain stable in the mean-square sense under an ETM of the form \eqref{eq:send:on:delta}, so we test it here under such noisy signals. We choose the consensus gain in the filter as $\kappa=500$, and the ETM parameters as $\sigma=1,\, \varepsilon=10^{-3}$. Since the consensus variables are vectors in this example, we consider the norms of the disagreement and event-triggered errors as inputs to the NN, which was trained using scalar values for these quantities. All estimates are randomly initialized, as the leaders' initial positions may be unknown to the followers.}

{Figure \ref{fig:estim} shows the state estimates of all follower agents together with the actual leader trajectories, showing that they successfully converge towards correct estimates despite their limited and noisy measurements and reduced communication rates. The dynamic variable $\eta_i(t)$ behaves consistently with the previous simple examples, taking larger values during the initial transient to reduce communication, although it is affected by the noisiness of the signals in this example. The numerical comparison for this experiment is given in Table \ref{tab:sims} (the MSE for the position estimates is given with respect to the centralized filter as a performance baseline), confirming that NN-ETM achieves significant communication reductions compared to the full communication case and a better trade-off between accuracy and frequency of communication than fixed-threshold ETMs.}

\begin{table*}[]
\centering
\caption{{Averages of mean squared errors (MSE) and triggered events per node over several simulations for different setups.}}
\label{tab:sims}
\resizebox{\textwidth}{!}{%
\begin{tabular}{c|cccc|cccccc|cc|}
\cline{2-13}
 & \multicolumn{4}{c|}{Linear consensus} & \multicolumn{6}{c|}{Nonlinear consensus} & \multicolumn{2}{c|}{\multirow{2}{*}{\begin{tabular}[c]{@{}c@{}}{Target tracking}\\ {example}\end{tabular}}} \\ \cline{2-11}
 & \multicolumn{2}{c|}{$N=3$} & \multicolumn{2}{c|}{$N=10$} & \multicolumn{2}{c|}{$N=3$} & \multicolumn{2}{c|}{$N=10$} & \multicolumn{2}{c|}{$N \in \{3,8,10\}$} & \multicolumn{2}{c|}{} \\ \cline{2-13} 
 & MSE & \multicolumn{1}{c|}{Events} & MSE & Events & MSE & \multicolumn{1}{c|}{Events} & MSE & \multicolumn{1}{c|}{Events} & MSE & Events & MSE & Events \\ \hline
\multicolumn{1}{|c|}{Full communication} & 0.1703 & \multicolumn{1}{c|}{5000} & 0.4350 & 5000 & 0.0525 & \multicolumn{1}{c|}{5000} & 0.5758 & \multicolumn{1}{c|}{5000} & 1.6354 & 3900 & 3.4094 & 5000 \\
\multicolumn{1}{|c|}{Fixed-threshold, $\eta_i = 0.25$} & 0.1737 & \multicolumn{1}{c|}{173.72} & 0.4387 & 541.39 & 0.0627 & \multicolumn{1}{c|}{273.45} & 0.5864 & \multicolumn{1}{c|}{665.30} & 1.6461 & 436.43 & 3.4113 & 1145.22 \\
\multicolumn{1}{|c|}{DETM} & - & \multicolumn{1}{c|}{-} & - & - & 0.0525 & \multicolumn{1}{c|}{4471.02} & 0.5758 & \multicolumn{1}{c|}{4870.59} & 1.6380 & 3723.24 & - & - \\ \hline
\multicolumn{1}{|c|}{NN-ETM, $\lambda=10^{-3}$} & 0.1727 & \multicolumn{1}{c|}{167.85} & 0.4384 & 463.02 & 0.0626 & \multicolumn{1}{c|}{269.60} & 0.5858 & \multicolumn{1}{c|}{645.17} & 1.6413 & 405.36 & 3.4107 & 1078.63 \\
\multicolumn{1}{|c|}{NN-ETM, $\lambda=10^{-2}$} & 0.1774 & \multicolumn{1}{c|}{124.12} & 0.4445 & 342.66 & 0.0742 & \multicolumn{1}{c|}{185.53} & 0.6009 & \multicolumn{1}{c|}{427.66} & 1.6547 & 268.79 & 3.4108 & 1052.63 \\
\multicolumn{1}{|c|}{NN-ETM, $\lambda=10^{-1}$} & 0.2337 & \multicolumn{1}{c|}{116.80} & 0.5080 & 264.72 & 0.1429 & \multicolumn{1}{c|}{100.08} & 0.6751 & \multicolumn{1}{c|}{230.59} & 1.7204 & 140.59 & 3.4125 & 1007.96 \\
\multicolumn{1}{|c|}{NN-ETM, $\lambda=1$} & 0.2536 & \multicolumn{1}{c|}{96.50} & 0.5387 & 230.33 & 0.1732 & \multicolumn{1}{c|}{78.37} & 0.7166 & \multicolumn{1}{c|}{199.51} & 1.7384 & 119.67 & 3.4138 & 963.93 \\ \hline
\end{tabular}%
}
\end{table*}

\begin{figure}
    \centering
    \includegraphics[width=\linewidth]{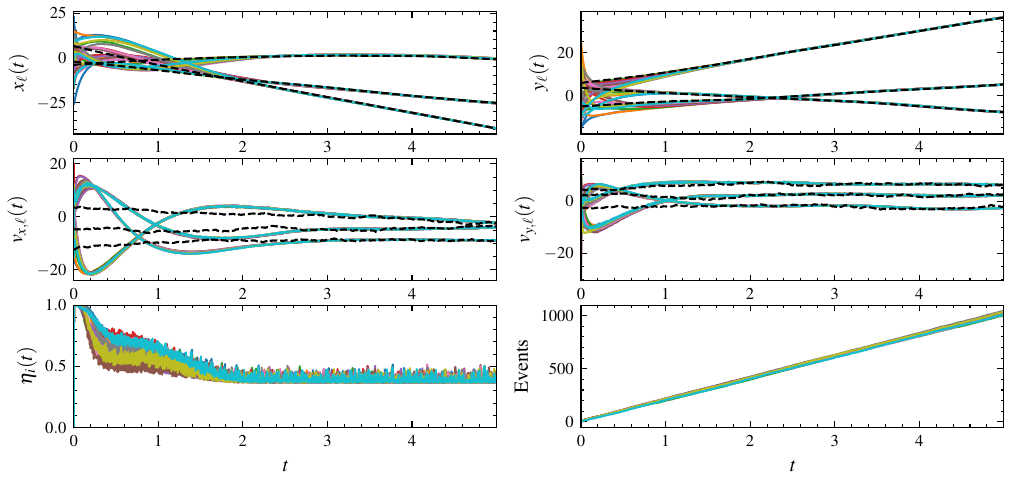}
    \caption{{Target tracking example, consensus-based observer with NN-ETM ($\lambda=10^{-2}$). All followers recover the leaders' position and velocity coordinates (black dashed lines).}}
    \label{fig:estim}
\end{figure}

\section{Conclusions}\label{sec:conclusions}

In this work, we have studied the problem of safely incorporating neural networks in ETMs for consensus problems.
We have proposed design criteria for event-triggered consensus, which decouple the ETM and consensus protocol design as separate blocks under mild constraints. Under said constraints, namely that the disturbance caused by the ETM is bounded and that the consensus protocol is ISpS under communication disturbances, we proposed a general ETM optimized by a neural network. The NN-ETM exploits the advantages of neural networks, offering a flexible alternative to hand-crafted ETMs while providing guarantees of boundedness for the consensus error. We have tested the idea in several simulation scenarios, including different consensus protocols, network sizes and topologies, as well as a collaborative target-tracking application, showing the potential advantages of NN-ETM against traditional ETMs. Testing other architectures and inputs for the NN to achieve different behaviors and further improve performance is left as future work.

\bibliographystyle{IEEEtran}
\bibliography{biblio}

\end{document}